\documentclass[10pt,journal,a4paper,twoside,twocolumn]{IEEEtran}
\usepackage{cite}
\usepackage{float}
\usepackage{stfloats}
\usepackage{fancyhdr}
\usepackage{color}
\usepackage[inline]{enumitem}
\usepackage{amsmath,amsfonts,amssymb}
\usepackage{graphicx}

\usepackage{algorithm}
\usepackage{algorithmic}
\usepackage{epsfig}
\usepackage{epstopdf}
\usepackage[caption=false]{subfig}
\usepackage{bm}
\usepackage{array}
\usepackage{bbding}
\usepackage{booktabs} 
\usepackage{multirow} 
\usepackage{pifont}

\usepackage{bbm}

\newtheorem{Lemma}{Lemma}
\newtheorem{Rem}{Remark}

\begin{document}
    \title{Shaping Delay-Doppler Ambiguity in Practical OFDM-ISAC}
	
	\author{\IEEEauthorblockN{
		Cheng Luo, \IEEEmembership{Member, IEEE}, Luping Xiang, \IEEEmembership{Senior Member, IEEE} and Kun Yang, \IEEEmembership{Fellow, IEEE}
		}
			\\
		\thanks{Cheng Luo, Luping Xiang and Kun Yang are with the State Key Laboratory of Novel Software Technology, Nanjing University, Nanjing, 210008, China, and also with the Institute of Intelligent Networks and Communications (NINE) and the School of Intelligent Software and Engineering, Nanjing University, Suzhou Campus, Suzhou 215163, China. (email: chengluo@nju.edu.cn; luping.xiang@nju.edu.cn; kunyang@nju.edu.cn).}
		\thanks{(Corresponding author: Luping Xiang.)}
	}

\maketitle
	
\thispagestyle{fancy} 
\lhead{} 
\chead{} 
\rhead{} 
\lfoot{} 
\cfoot{} 
\rfoot{\thepage} 
\renewcommand{\headrulewidth}{0pt} 
\renewcommand{\footrulewidth}{0pt} 
\pagestyle{fancy}

\rfoot{\thepage} 

\begin{abstract}
Orthogonal frequency-division multiplexing (OFDM) is a key waveform for integrated sensing and communication (ISAC). Existing OFDM ambiguity analyses, however, typically assume fully occupied data-only waveforms, whereas practical frames contain direct-current and edge-guard nulls, fixed pilots, and random payload symbols. This mixed resource structure reshapes the self-ambiguity function and induces prominent sidelobes in the sensing region of interest (ROI). We therefore propose ROI-oriented deep block-unitary precoded OFDM (DBU-OFDM), which combines resource-specific trainable unitary transformations with dedicated sensing subcarriers to suppress ROI sidelobes while preserving the prescribed resource support. We develop a constraint-preserving parameterization capable of representing arbitrary unitary matrices and an ROI-aware sensing-support initialization. We prove periodic autocorrelation function (P-ACF) invariance under phase-only optimization and zero-delay Doppler-cut invariance under arbitrary unitary transformations. Numerical results corroborate the P-ACF optimality of cyclic-prefix OFDM (CP-OFDM) and provide numerical support for the conjecture that conventional OFDM is globally optimal in the considered aperiodic autocorrelation function (A-ACF) setting. Unitary pilot optimization improves the ROI peak-to-sidelobe ratio by over 2 dB, while combining it with dedicated sensing subcarriers yields gains ranging from several to tens of dB. Under the considered configuration, the A-ACF mode also provides substantially greater ROI ambiguity-shaping capability than the P-ACF mode.
\end{abstract}


\section{Introduction}\label{sec:I}

\IEEEPARstart{I}{ntegrated} sensing and communication (ISAC) is an important capability of sixth-generation wireless systems, where the wireless infrastructure is expected to provide communication together with environmental sensing, localization, and tracking. By sharing spectrum, hardware, and transmitted signals, ISAC supports these functions within a unified wireless platform \cite{wang2023road,liu2022integrated,cui2021integrating,chiriyath2017radar}. Its waveform and signal-processing foundations have therefore received extensive attention \cite{zhou2022waveform,sturm2011waveform}.

Among the candidate ISAC waveforms, orthogonal frequency-division multiplexing (OFDM) is attractive because of its spectral efficiency, flexible resource allocation, and mature fast Fourier transform (FFT)-based transceiver architecture. Early OFDM sensing studies established joint radar-communication transmission and derived range and velocity estimation from the received OFDM symbols \cite{sturm2009ofdm,braun2010maximum}. The corresponding ambiguity function characterizes the delay-Doppler resolution and sidelobe structure \cite{he2012waveform}.

Building on OFDM, subsequent ISAC studies have investigated beamforming, transceiver design, and time-frequency resource allocation. For beamforming design, Liu et al. designed dual-functional multiple-input multiple-output waveforms (MIMO) by jointly considering the radar beampattern and downlink multiuser interference \cite{liu2018toward}. Their subsequent study summarized the principal radar-communication design approaches and developed a hybrid analog-digital transceiver for joint target search and communication-channel estimation \cite{liu2020joint}. Wei et al. formulated communication, sensing, and joint MIMO-OFDM waveform designs using mutual information \cite{wei2023waveform}. Moreover, He et al. projected an additional sensing signal onto the communication-channel null space and optimized its autocorrelation sidelobes \cite{he2024dual}.

In addition to beamforming and transceiver design, the OFDM time-frequency resources provide further waveform variables for ISAC. Bica et al. investigated multicarrier waveform optimization for radar-communication convergence \cite{bica2019multicarrier}. Shi et al. jointly allocated subcarriers and power \cite{shi2019joint} and further considered subcarrier selection in multicarrier dual-functional systems \cite{shi2021joint}. Chen et al. optimized subcarrier placement and power according to ranging sidelobes, communication rate, and range resolution \cite{chen2023joint}. This resource-allocation framework was extended to channel uncertainty and sensing quality-of-service requirements in \cite{cao2023robust,dong2023sensing}. Other studies optimized the phase perturbations \cite{zhou2019joint}, time-frequency power coefficients under limited feedforward \cite{keskin2021limited}, and the ambiguity cells within a prescribed region of interest (ROI) \cite{zhang2024cross}. 

The sensing optimality of OFDM has also been studied analytically. Under quadrature amplitude modulation (QAM) and phase-shift keying (PSK), Liu et al. proved that cyclic-prefix OFDM (CP-OFDM) minimizes the expected periodic autocorrelation function (P-ACF) sidelobe power at every nonzero delay and is a local optimum for the expected integrated aperiodic autocorrelation function (A-ACF) sidelobe power \cite{liu2025cp}. Zhang et al. extended this analysis to discrete two-dimensional ambiguity functions and showed that the optimal basis depends on the ambiguity definition and symbol distribution \cite{zhang2025discreteaf}.

As the waveform and receiver variables become increasingly coupled, learning-based methods provide an alternative means of optimizing ISAC systems. Cammerer et al. developed trainable communication systems and demonstrated their practical implementation \cite{cammerer2020trainable}. Ait Aoudia et al. jointly learned the transmit and receive filters, constellation, labeling, and detector under communication waveform constraints \cite{ait2022waveform}. Qi et al. jointly learned an uplink sensing waveform and communication receive beamforming \cite{qi2024deep}. Moreover, Jiang et al. unfolded passive sensing, signal detection, and channel reconstruction into a model-driven network \cite{jiang2024isac}, while Mateos-Ramos et al. developed model-based end-to-end learning for multi-target ISAC under hardware impairments \cite{mateos2025model}.

Against this background, three coupled challenges remain. First, many analytical studies focus on one-dimensional correlation functions. Although the discrete two-dimensional ambiguity function has also been analyzed, directly controlling the largest coherent sidelobe within a finite delay-Doppler ROI remains difficult. Second, the fully occupied resource structures used in the theoretical analyses differ from practical OFDM frames containing fixed pilots, edge and direct-current (DC) nulls, and random payload symbols, as shown in Fig.~\ref{fig:OFDMgframe}. These prescribed resources reshape the sidelobes and introduce frame-dependent variations. Third, the sensing-resource budget and support locations determine the waveform-shaping freedom, delay-correlation structure, and remaining communication payload. Therefore, the sensing support, sensing symbols, and continuous waveform transformation need to be considered within a unified formulation.

In our preceding work, deep block-unitary precoded OFDM (DBU-OFDM) introduced trainable block-unitary transformations within prescribed resource groups while preserving resource isolation and traditional communication processing \cite{luo2026dbu}. In this paper, we extend this structure to two-dimensional ROI ambiguity shaping over the complete practical OFDM resource configuration. Specifically, the resource-block transformations, sensing-subcarrier support, and sensing symbols are jointly optimized to suppress the largest self-ambiguity sidelobe within the selected delay-Doppler ROI while maintaining the prescribed resource structure.

The main contributions of this paper are summarized as follows.
\begin{itemize}
    \item \textbf{Two-dimensional ROI-oriented sensing formulation.} We establish P-ACF and A-ACF models for CP-OFDM and zero-guard operation and formulate the largest normalized sidelobe over a prescribed delay-Doppler ROI. 

    \item \textbf{Resource-structured DBU-OFDM waveform design.} We develop a DBU-OFDM waveform that combines resource-block unitary transformations with dedicated sensing-subcarrier allocation, thereby providing additional degrees of freedom for sensing and ROI sidelobe shaping. The construction preserves resource isolation, signal energy, exact invertibility, and DFT-based single-tap equalization under the stated conditions. Moreover, we rigorously establish the complete representation condition $K_{\vartheta}\geq N_{\vartheta}-1$, under which the Householder parameterization represents any $N_{\vartheta}\times N_{\vartheta}$ unitary resource block exactly.

    \item \textbf{Ambiguity-guided sensing-subcarrier allocation.} We decompose the complete self-ambiguity function into the sensing/pilot, communication, and cross-ambiguity terms. Based on the support-dependent sensing/pilot term, we formulate an array-factor criterion that provides an informed and ROI-aware sensing-support initialization. The support is then refined jointly with the sensing symbols and unitary resource blocks using the complete ambiguity loss, which effectively suppresses the sidelobes within the selected ROI.

    \item \textbf{Analytical properties and numerical insights.} We rigorously prove the P-ACF invariance under phase-only design and the zero-delay Doppler-cut invariance under arbitrary DBU-OFDM transformations. Under the fully occupied resource structure, the learned transformations recover an OFDM-equivalent monomial form. This result numerically confirms the established P-ACF optimality of CP-OFDM and provides evidence that conventional OFDM may also attain the global optimum of the considered zero-Doppler A-ACF delay-slice PSLR problem within the unitary waveform class. Moreover, the numerical results show that A-ACF provides greater ROI sidelobe-shaping freedom than P-ACF under the considered configuration.
\end{itemize}

The remainder of this paper is organized as follows. Section~\ref{sec:II} establishes the OFDM signal model and formulates the ROI-oriented sidelobe suppression problem. Section~\ref{sec:III} develops the resource-structured DBU-OFDM waveform and sensing-subcarrier allocation. Section~\ref{sec:IV} presents the numerical results. Finally, Section~\ref{sec:conclusion} concludes the paper.

\emph{Notation:} For any positive integer $Q$, we define $\mathbb{Z}_{Q}\triangleq\{0,\ldots,Q-1\}$. For any integer $a$, $a\bmod Q$ denotes the nonnegative remainder of $a$ upon division by $Q$ and therefore belongs to $\mathbb{Z}_{Q}$. The notations $[\cdot]_a$ and $[\cdot]_{a,b}$ denote the $a$-th entry of a vector and the $(a,b)$-th entry of a matrix, respectively. For a matrix $\mathbf{A}$, $[\mathbf{A}]_{\bm{\mathcal{I}},\bm{\mathcal{J}}}$ denotes the submatrix indexed by the row set $\bm{\mathcal{I}}$ and column set $\bm{\mathcal{J}}$. The operators $(\cdot)^T$, $(\cdot)^H$, and $(\cdot)^*$ denote the transpose, conjugate transpose, and complex conjugate, respectively. The notation $|\cdot|$ denotes scalar magnitude or set cardinality, according to its argument, while $\|\cdot\|_F$ and $\mathbb{E}[\cdot]$ denote the Frobenius norm and expectation, respectively. The imaginary unit is denoted by $j\triangleq\sqrt{-1}$.

\begin{figure}
	\centering
	\includegraphics[width=0.85\linewidth]{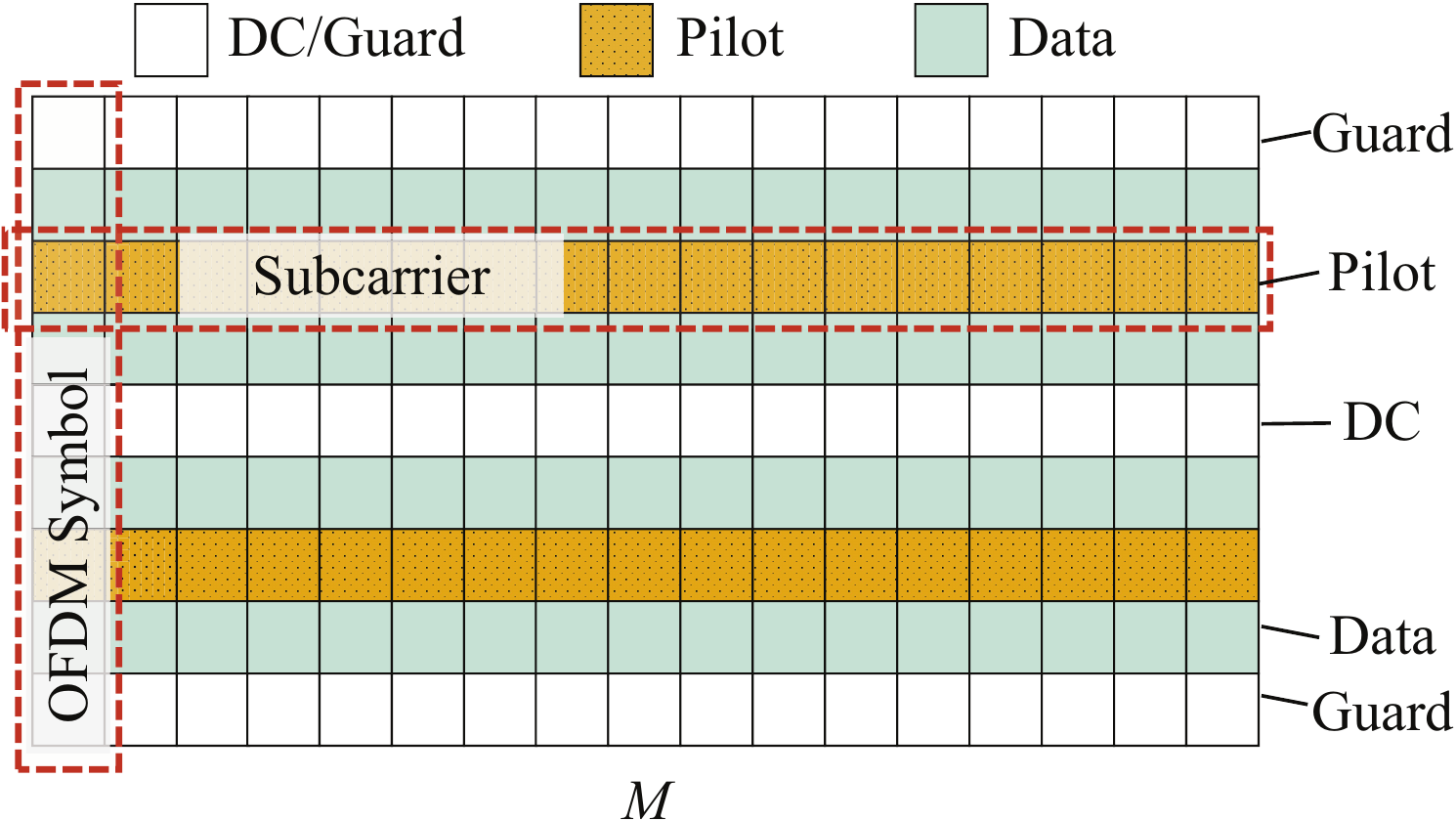}
	\caption{Considered OFDM system with comb-type pilots.}
	\label{fig:OFDMgframe}
\end{figure}

\section{System Model and Problem Formulation}\label{sec:II}
In this section, we establish the OFDM signal model, characterize its delay-Doppler ambiguity under P-ACF and A-ACF, and formulate the ROI-oriented sidelobe suppression problem.

\subsection{OFDM Frame and Resource Model}\label{sec:II-A}

We consider a conventional OFDM system in which each frame contains $M$ consecutive OFDM symbols and each OFDM symbol occupies $N$ subcarriers with comb-type pilots, as shown in Fig. \ref{fig:OFDMgframe}. The subcarrier and OFDM-symbol indices are denoted by $n\in\mathbb{Z}_{N}\triangleq\{0,\ldots,N-1\}$ and $m\in\mathbb{Z}_{M}\triangleq\{0,\ldots,M-1\}$, respectively. The physical subcarriers are partitioned into the data set $\bm{\mathcal{K}}_{d}$, the pilot set $\bm{\mathcal{K}}_{p}$, and the null set $\bm{\mathcal{K}}_{0}$, satisfying
\begin{equation}
\begin{aligned}
    &\bm{\mathcal{K}}_{d}\cup\bm{\mathcal{K}}_{p}\cup\bm{\mathcal{K}}_{0}
     =\mathbb{Z}_{N},\\
    &\bm{\mathcal{K}}_{d}\cap\bm{\mathcal{K}}_{p}
     =\bm{\mathcal{K}}_{d}\cap\bm{\mathcal{K}}_{0}
      =\bm{\mathcal{K}}_{p}\cap\bm{\mathcal{K}}_{0}
      =\varnothing.
\end{aligned}                                                     \label{eq:resource_partition}
\end{equation}
The null set contains $N_g$ guard subcarriers at each band edge and $N_{dc}$ subcarriers around direct current (DC), whereas $N_p=|\bm{\mathcal{K}}_{p}|$ comb-type pilots are used for channel estimation. Thus, the numbers of null and data subcarriers are
\begin{align}
   &|\bm{\mathcal{K}}_{0}|= N_0=2N_g+N_{dc},\nonumber\\
   &|\bm{\mathcal{K}}_{d}|= N_d=N-N_p-N_0.                    \label{eq:resource_cardinalities}
\end{align}

Let $S_{n,m}$ denote the symbol carried by the $n$-th subcarrier of the $m$-th OFDM symbol, and define the frequency-domain OFDM symbol vector as
\begin{align}
    \mathbf{s}_m
    =[S_{0,m},\ldots,S_{N-1,m}]^{T}\in\mathbb{C}^{N\times 1}.
                                                                    \label{eq:ofdm_symbol_vector}
\end{align}
The frequency-domain frame is then expressed as 
\begin{align}
    \mathbf{S}
    =[\mathbf{s}_0,\ldots,\mathbf{s}_{M-1}]
    \in\mathbb{C}^{N\times M},                                  \label{eq:freq_frame}
\end{align}
and its $(n,m)$-th entry is given by
\begin{align}
    &S_{n,m}
    =
    \begin{cases}
        D_{n,m}, & n\in\bm{\mathcal{K}}_{d},\\
        P_{n,m}, & n\in\bm{\mathcal{K}}_{p},\\
        0,       & n\in\bm{\mathcal{K}}_{0},
    \end{cases}                                                     \label{eq:resource_symbols}
\end{align}
$D_{n,m}$ is drawn from a QAM constellation normalized to unit average energy, $P_{n,m}$ is a known pilot symbol, and the entries indexed by $\bm{\mathcal{K}}_{0}$ are zero.

Let $\Delta f$ denote the subcarrier spacing. The sampling interval and useful OFDM-symbol duration are $T_s=1/(N\Delta f)$ and $T_u=NT_s=1/\Delta f$, respectively. Define the normalized $N$-point discrete Fourier transform (DFT) matrix by $[\mathbf{F}_{N}]_{a,b}=N^{-1/2}e^{-j2\pi ab/N}$ for $a,b\in\mathbb{Z}_{N}$. The useful time-domain signal of the $m$-th OFDM symbol is
\begin{align}
    \mathbf{x}_m=\mathbf{F}_{N}^{H}\mathbf{s}_m
    \in\mathbb{C}^{N\times 1}.                                  \label{eq:ofdm_modulation}
\end{align}
Collecting the $M$ time-domain OFDM symbols gives
\begin{align}
    \mathbf{X}
    =[\mathbf{x}_0,\ldots,\mathbf{x}_{M-1}]
    =\mathbf{F}_{N}^{H}\mathbf{S}
    \in\mathbb{C}^{N\times M}.                                  \label{eq:ofdm_frame}
\end{align}
The $n$-th useful time-domain sample of the $m$-th OFDM symbol is denoted by $x_{n,m}=[\mathbf{x}_m]_n$, where $n\in\mathbb{Z}_{N}$. We formulate the sensing model over these $N$ useful samples. The cyclic prefix (CP) produces a cyclic extension of the useful block, whereas the zero guard produces a zero extension.

\subsection{Operating-Mode-Dependent Self-Ambiguity Functions}\label{sec:II-B}

To characterize the delay-domain sidelobes, we first introduce the one-dimensional periodic autocorrelation function (P-ACF) and aperiodic autocorrelation function (A-ACF) \cite{liu2025cp}. For the $m$-th time-domain OFDM symbol $\mathbf{x}_m$, their values at a nonnegative delay index $k\in\mathbb{Z}_N$ are defined as
\begin{align}
    &r_{\mathrm{P},m}[k]
    =\sum_{n=0}^{N-1}x_{n,m}^{*}x_{(n+k)\bmod N,m},          \label{eq:p_acf_1d}\\
    &r_{\mathrm{A},m}[k]
    =\sum_{n=0}^{N-1-k}x_{n,m}^{*}x_{n+k,m}.                 \label{eq:a_acf_1d}
\end{align}
The P-ACF in Eq. \eqref{eq:p_acf_1d} circularly wraps the delayed sample index and therefore contains $N$ sample pairs, whereas the A-ACF in Eq. \eqref{eq:a_acf_1d} contains only the $N-k$ overlapping sample pairs.

\begin{figure}
	\centering
	\includegraphics[width=0.998\linewidth]{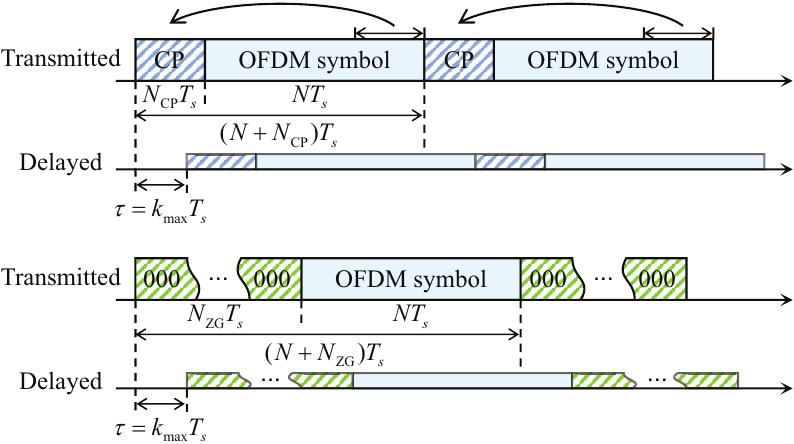}
	    \vspace{-10pt}
    {
	\caption{Representative physical waveform scenarios corresponding to the P-ACF and A-ACF models. The upper diagram shows continuously transmitted CP-OFDM blocks and their delayed replicas, whereas the lower diagram shows CP-free OFDM pulses separated by zero guards and their delayed replicas.}
	\label{fig:OFDMZGCP}
    }

\end{figure}

The P-ACF and A-ACF definitions correspond to two representative physical waveform scenarios, as illustrated in Fig.~\ref{fig:OFDMZGCP}. The P-ACF in Eq. \eqref{eq:p_acf_1d} corresponds to continuously transmitted CP-OFDM blocks. Let $N_{\mathrm{CP}}$ denote the CP length in samples. The resulting block repetition interval is $T_{\mathrm{P}}=(N+N_{\mathrm{CP}})T_s$. When the largest modeled round-trip delay does not exceed $N_{\mathrm{CP}}$ samples, CP removal converts each delayed useful block into a circular shift of the transmitted block, which agrees with the periodic indexing in Eq. \eqref{eq:p_acf_1d}. The A-ACF in Eq. \eqref{eq:a_acf_1d} corresponds to CP-free OFDM pulses separated by time-domain zero guards.\footnote{CP-OFDM is generally preferred when its CP covers the sensing delay range. For a much larger range, extending the CP incurs substantial time and transmit-energy overhead. CP-free OFDM with a zero guard avoids transmitting redundant samples during the guard interval, at the cost of a lower duty cycle.} Let $N_{\mathrm{ZG}}$ denote the zero-guard length in samples, which gives the pulse repetition interval $T_{\mathrm{A}}=(N+N_{\mathrm{ZG}})T_s$. When the zero guard covers the largest modeled round-trip delay, each return is represented by a linear shift of the zero-extended useful block without overlapping the subsequent pulse. This finite-overlap structure agrees with the aperiodic summation in Eq. \eqref{eq:a_acf_1d}. Increasing $N_{\mathrm{ZG}}$ accommodates a larger delay range at the cost of a lower transmission duty cycle.

Let $T_{\chi}$ denote $T_{\mathrm{P}}$ or $T_{\mathrm{A}}$ for mode $\chi\in\{\mathrm{P},\mathrm{A}\}$, respectively, and define the centered Doppler-index set $\mathbb{Z}_M^{\mathrm{c}}\triangleq\{-\lfloor M/2\rfloor,\ldots,\lceil M/2\rceil-1\}$. Let $k_{\max}\in\mathbb{Z}_N$ denote the largest round-trip delay index within the modeled sensing range. We restrict the modeled delays to integer sample shifts, neglect range migration over the $M$ blocks, and require $k_{\max}\leq N_{\mathrm{CP}}$ in the P-ACF mode or $k_{\max}\leq N_{\mathrm{ZG}}$ in the A-ACF mode. Let $q_{\max}^{(\chi)}\in\{0,\ldots,\lfloor M/2\rfloor\}$ denote the largest absolute Doppler index within this range. On the $M$-point slow-time grid, the corresponding maximum absolute physical Doppler frequency is $f_{\mathrm{D},\max}^{(\chi)}=q_{\max}^{(\chi)}/(MT_{\chi})$. We impose the worst-case small-intrasymbol-Doppler condition $\phi_{\mathrm{intra},\max}^{(\chi)}\triangleq2\pi f_{\mathrm{D},\max}^{(\chi)}T_u\ll1$, which becomes
\begin{align}
    &\phi_{\mathrm{intra},\max}^{(\mathrm{P})}=
      \dfrac{2\pi q_{\max}^{(\mathrm{P})}N}
            {M(N+N_{\mathrm{CP}})}\ll1, \nonumber\\
    &\phi_{\mathrm{intra},\max}^{(\mathrm{A})}=
      \dfrac{2\pi q_{\max}^{(\mathrm{A})}N}
            {M(N+N_{\mathrm{ZG}})}\ll1.                    \label{eq:small_doppler}
\end{align}
This condition gives $e^{j2\pi f_{\mathrm{D}}nT_s}\simeq1$ for every modeled $|f_{\mathrm{D}}|\leq f_{\mathrm{D},\max}^{(\chi)}$, so the Doppler phase is treated as constant over each useful block and as evolving only across the slow-time blocks. The intrasymbol Doppler and resulting intercarrier interference (ICI) are therefore negligible.

Let $r_{\chi,m}[k]$ denotes the corresponding per-symbol ACF. Coherently combining these ACFs across the $M$ OFDM symbols gives the two-dimensional self-ambiguity function{\footnote{The two-dimensional P-ACF ambiguity function defined in this subsection corresponds to the periodic fast-slow-time ambiguity function (FST-AF) in \cite{zhang2025discreteaf}. The two-dimensional A-ACF ambiguity function is its natural aperiodic extension.}}
\begin{align}
    &\Gamma_{\chi}[k,q]
    =\sum_{m=0}^{M-1}r_{\chi,m}[k]e^{-j2\pi qm/M}, \nonumber\\
    &(k,q)\in\mathbb{Z}_N\times\mathbb{Z}_M^{\mathrm{c}}, \quad\chi\in\{\mathrm{P},\mathrm{A}\}          \label{eq:self_af}
\end{align}

The indices $k$ and $q$ in Eq. \eqref{eq:self_af} denote the relative discrete delay and centered Doppler lags, respectively. Their corresponding physical offsets are $kT_s$ and $q/(MT_{\chi})$. The relative Doppler index $q$ is distinct from $q_{\max}^{(\chi)}$, which bounds the absolute physical Doppler range used to validate the small-intrasymbol-Doppler approximation. In both operating modes, the zero-delay, zero-Doppler value equals the waveform energy
\begin{align}
    &\Gamma_{\chi}[0,0]
    =\sum_{m=0}^{M-1}\sum_{n=0}^{N-1}|x_{n,m}|^2
    =\|\mathbf{X}\|_{F}^{2},  \quad\chi\in\{\mathrm{P},\mathrm{A}\}.                            \label{eq:af_origin}
\end{align}

Note that matched filtering maps each target to a translated replica of the self-ambiguity kernel. The Doppler translation is circular modulo $M$, whereas the delay translation is circular modulo $N$ in the P-ACF mode and remains locally valid within the finite correlation window in the A-ACF mode. For multiple targets, the superposition of these replicas implies that suppressing the kernel sidelobes within a prescribed ROI reduces the masking of weaker targets by stronger-target sidelobes. The sensing objective can therefore be formulated solely in terms of the self-ambiguity kernel without explicitly introducing individual target parameters.
\subsection{ROI-Oriented PSLR Objective}\label{sec:II-C}

The waveform design controls the largest self-ambiguity sidelobe within a prescribed delay-Doppler region of interest (ROI). Let $\Omega_{\mathrm{ROI}}\subseteq\mathbb{Z}_N\times\mathbb{Z}_M^{\mathrm{c}}$ denote the set of relative delay-Doppler index pairs $(k,q)$ selected within the sensing range. The corresponding sidelobe evaluation set is
\begin{align}
    &\Omega_{\mathrm{SL}}
    =\Omega_{\mathrm{ROI}}\setminus\{(0,0)\},
    \qquad \Omega_{\mathrm{SL}}\neq\varnothing.                 \label{eq:roi_sets}
\end{align}
If $(0,0)\notin\Omega_{\mathrm{ROI}}$, then we have $\Omega_{\mathrm{SL}}=\Omega_{\mathrm{ROI}}$. For such an ROI, we still use $\Gamma_{\chi}[0,0]$ as the mainlobe reference.

Using the positive mainlobe-to-sidelobe convention, the ROI peak-to-sidelobe ratio (PSLR) is defined as
\begin{align}
    &\mathrm{PSLR}_{\mathrm{ROI}}^{(\chi)}(\mathbf{X})
    =20\log_{10}
      \frac{|\Gamma_{\chi}[0,0]|}
           {\displaystyle
            \max_{(k,q)\in\Omega_{\mathrm{SL}}}
            |\Gamma_{\chi}[k,q]|}.              \label{eq:roi_pslr}
\end{align}

A larger value of the PSLR in Eq. \eqref{eq:roi_pslr} indicates stronger suppression of the largest self-ambiguity sidelobe within the selected ROI.

Note that maximizing the PSLR in Eq. \eqref{eq:roi_pslr} is equivalent to minimizing the largest normalized sidelobe power over $\Omega_{\mathrm{SL}}$. Since the maximum operation over $\Omega_{\mathrm{SL}}$ is nonsmooth, we replace it with a log-sum-exp approximation parameterized by $\alpha>0$. The resulting differentiable loss is given by
\begin{align}
    &\mathcal{L}_{\chi}(\mathbf{X})
    =\frac{1}{\alpha}\log
      \left[
        \sum_{(k,q)\in\Omega_{\mathrm{SL}}}
        \exp\!\left(
          \alpha\frac{|\Gamma_{\chi}[k,q]|^{2}}
          {|\Gamma_{\chi}[0,0]|^{2}}
        \right)
      \right].                                                   \label{eq:roi_loss}
\end{align}

As $\alpha$ increases, Eq. \eqref{eq:roi_loss} approaches the maximum normalized sidelobe power over $\Omega_{\mathrm{SL}}$. Minimizing this loss therefore provides a smooth approximation to maximizing the PSLR in Eq. \eqref{eq:roi_pslr}.

\section{Proposed ROI-Oriented DBU-OFDM}\label{sec:III}
\subsection{Impact of Practical OFDM Resource Structures on Self-Ambiguity}\label{sec:III-A}

Under QAM/PSK signaling, OFDM minimizes the expected P-ACF sidelobes and locally minimizes the expected integrated A-ACF sidelobes\cite{liu2025cp}. These results characterize the average delay-domain performance under a fully occupied OFDM resource structure with random data symbols. However, practical null resources alter the occupied spectral support, while fixed pilots introduce a deterministic pattern among the random communication symbols. These resource constraints can therefore reshape the two-dimensional self-ambiguity surface. Fig.~\ref{fig:acf_resource_effect} compares the P-ACF and A-ACF modes with and without these fixed resources.

\begin{figure*}[!t]
    \centering
    \captionsetup[subfloat]{width=0.2\textwidth,justification=centering,singlelinecheck=false}
    \subfloat[P-ACF with all subcarriers carrying data.]{%
        \includegraphics[height=0.25\textwidth]{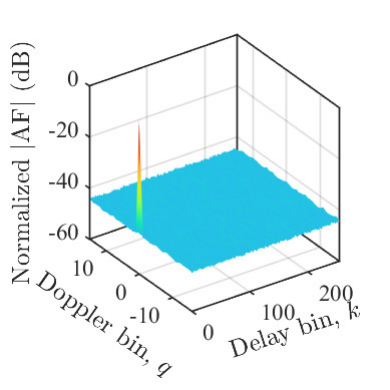}%
        \label{fig:pacf_full_data}}
      \hspace{-50pt}
    \hfill
    \subfloat[P-ACF with null and pilot subcarriers.]{%
        \includegraphics[height=0.25\textwidth]{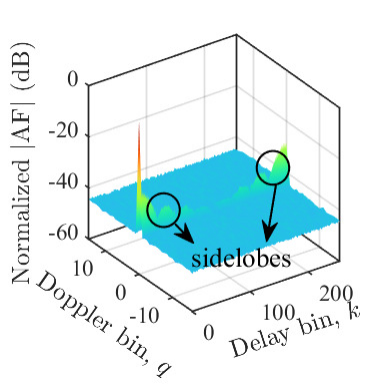}%
        \label{fig:pacf_null_pilot}}
    \hfill
    \hspace{-50pt}
    \subfloat[A-ACF with all subcarriers carrying data.]{%
        \includegraphics[height=0.25\textwidth]{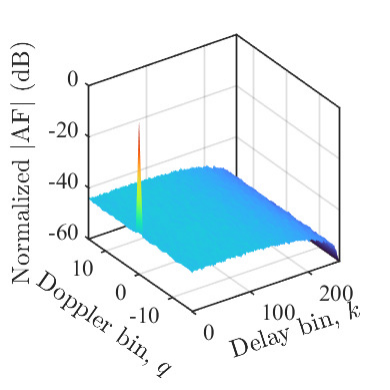}%
        \label{fig:aacf_full_data}}
    \hfill
    \hspace{-50pt}
    \subfloat[A-ACF with null and pilot subcarriers.]{%
        \includegraphics[height=0.25\textwidth]{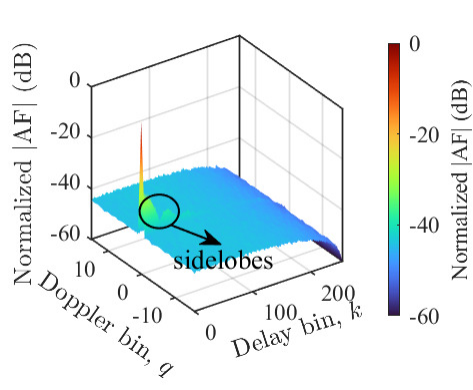}%
        \label{fig:aacf_null_pilot}}
    \caption{Normalized two-dimensional self-ambiguity magnitudes of conventional OFDM under the P-ACF and A-ACF modes with different OFDM resource structures. Each surface shows $20\log_{10}\!\left(|\Gamma_{\chi}[k,q]|/|\Gamma_{\chi}[0,0]|\right)$.}
    \label{fig:acf_resource_effect}
\end{figure*}

Specifically, when all subcarriers carry data, both correlation modes exhibit a dominant mainlobe and a comparatively smooth sidelobe surface. The P-ACF in Eq. \eqref{eq:p_acf_1d} retains circular overlap over the entire delay interval and therefore has a more uniform sidelobe floor along the delay dimension. In contrast, the A-ACF surface decreases toward large delay indices because the number of overlapping sample pairs in Eq. \eqref{eq:a_acf_1d} decreases with the delay.

Moreover, introducing the edge and DC nulls together with the comb-type pilots produces delay-structured sidelobe ridges around zero Doppler in both correlation modes, as shown in Fig.~\ref{fig:pacf_null_pilot} and Fig.~\ref{fig:aacf_null_pilot}. This effect is more pronounced in the P-ACF example, where elevated sidelobes extend to larger delay indices. The comparison demonstrates that the prescribed subcarrier resources can substantially reshape the two-dimensional self-ambiguity function. Thus, an ROI-oriented waveform design must account for the complete resource-constrained frame and directly shape its self-ambiguity sidelobes over the selected ROI.

\subsection{Resource-Structured DBU-OFDM Architecture}\label{sec:III-B}

\subsubsection{Resource-Structured Block-Unitary Construction}

The comparison in Fig.~\ref{fig:acf_resource_effect} shows that fixed pilot and null resources reshape the two-dimensional self-ambiguity function and produce elevated sidelobes over parts of the delay-Doppler plane. Since their locations are prescribed, the waveform-shaping operation needs to preserve the subcarrier support, resource isolation, signal energy, exact invertibility, and conventional DFT-based processing. To meet these requirements, DBU-OFDM applies a trainable frequency-domain transformation before the inverse DFT (IDFT). This construction reuses the conventional resource-mapping and CP/ZG processing modules. When the CP covers the channel delay spread and the intrasymbol channel variation and intercarrier interference (ICI) are negligible, DBU-OFDM also retains DFT-based single-tap equalization.

According to Eq. \eqref{eq:resource_symbols}, we collect the data symbols over $\bm{\mathcal{K}}_{d}$ and the pilot symbols over $\bm{\mathcal{K}}_{p}$ in physical subcarrier order. The corresponding resource subvectors are given by
\begin{align}
    &\mathbf{s}_{d,m}
    =[\mathbf{s}_{m}]_{\bm{\mathcal{K}}_{d}}
    \in\mathbb{C}^{N_d\times 1},\qquad
    \mathbf{s}_{p,m}
    =[\mathbf{s}_{m}]_{\bm{\mathcal{K}}_{p}}
    \in\mathbb{C}^{N_p\times 1}.                              \label{eq:resource_subvectors}
\end{align}

We then introduce a fixed binary permutation matrix $\mathbf{\Pi}\in\{0,1\}^{N\times N}$ that maps the resource-grouped order to the physical subcarrier order. It satisfies $\mathbf{\Pi}^{T}\mathbf{\Pi}=\mathbf{\Pi}\mathbf{\Pi}^{T}=\mathbf{I}_{N}$. Its transpose maps the interleaved physical resources into the grouped vector, which is given by
\begin{align}
    &\mathbf{s}_{g,m}
    =\mathbf{\Pi}^{T}\mathbf{s}_{m}
    =\begin{bmatrix}
        \mathbf{s}_{d,m}^{T} &
        \mathbf{s}_{p,m}^{T} &
        \mathbf{0}_{N_0}^{T}
      \end{bmatrix}^{T}.                                      \label{eq:grouped_resource_vector}
\end{align}
The permutation only reorders the entries and therefore preserves their values and energy. Let $\mathbf{U}_{d}\in\mathbb{C}^{N_d\times N_d}$, $\mathbf{U}_{p}\in\mathbb{C}^{N_p\times N_p}$, and $\mathbf{U}_{0}\in\mathbb{C}^{N_0\times N_0}$ denote the transformations within the data, pilot, and null subspaces, respectively. The grouped transformation is
\begin{align}
    &\mathbf{U}_{g}
    =\operatorname{blkdiag}\!\left(
        \mathbf{U}_{d},\mathbf{U}_{p},\mathbf{U}_{0}
      \right).                                                  \label{eq:grouped_unitary}
\end{align}
To preserve signal energy and exact invertibility within the occupied subspaces while bypassing the null resources, we impose
\begin{align}
    &\mathbf{U}_{d}^{H}\mathbf{U}_{d}=\mathbf{I}_{N_d},\qquad
      \mathbf{U}_{p}^{H}\mathbf{U}_{p}=\mathbf{I}_{N_p},\qquad
      \mathbf{U}_{0}=\mathbf{I}_{N_0}.                           \label{eq:resource_block_constraints}
\end{align}
These conditions give $\mathbf{U}_{g}^{H}\mathbf{U}_{g}=\mathbf{I}_{N}$, while the block-diagonal form prevents mixing among the data, pilot, and null resources. Mapping the grouped transformation back to the physical subcarrier order gives
\begin{align}
    &\mathbf{U}
    =\mathbf{\Pi}\mathbf{U}_{g}\mathbf{\Pi}^{T}.              \label{eq:full_dbu_unitary}
\end{align}
Applying this transformation to the complete frame gives the physical-order frequency-domain frame $\bar{\mathbf{S}}\in\mathbb{C}^{N\times M}$ as
\begin{align}
    &\bar{\mathbf{S}}
    =\mathbf{U}\mathbf{S}
    =\mathbf{\Pi}\mathbf{U}_{g}\mathbf{\Pi}^{T}\mathbf{S}. \label{eq:dbu_frequency_frame}
\end{align}
The IDFT then produces the transmitted useful waveform $\bar{\mathbf{X}}\in\mathbb{C}^{N\times M}$ according to
\begin{align}
    &\bar{\mathbf{X}}
    =\mathbf{F}_{N}^{H}\bar{\mathbf{S}}
    =\mathbf{F}_{N}^{H}
      \mathbf{\Pi}\mathbf{U}_{g}\mathbf{\Pi}^{T}\mathbf{S}.    \label{eq:dbu_frame}
\end{align}
Each column of $\bar{\mathbf{X}}$ is a useful time-domain DBU-OFDM block before CP or ZG insertion. Let $\bar{\mathbf{s}}_m$ denote the $m$-th column of $\bar{\mathbf{S}}$. Its grouped resource vector is
\begin{align}
    &\mathbf{\Pi}^{T}\bar{\mathbf{s}}_{m}
    =\begin{bmatrix}
        (\mathbf{U}_{d}\mathbf{s}_{d,m})^{T} &
        (\mathbf{U}_{p}\mathbf{s}_{p,m})^{T} &
        \mathbf{0}_{N_0}^{T}
      \end{bmatrix}^{T}.                                      \label{eq:dbu_resource_action}
\end{align}
Thus, the transformed data and effective pilots remain on their prescribed subcarriers, while the DC and edge-guard subcarriers remain zero. The full-size transformation is unitary since
\begin{align}
    &\mathbf{U}^{H}\mathbf{U}
    =\mathbf{\Pi}\mathbf{U}_{g}^{H}\mathbf{U}_{g}\mathbf{\Pi}^{T}
    =\mathbf{I}_{N}.                                           \label{eq:dbu_unitarity}
\end{align}
Consequently, the transformation preserves signal energy and is exactly inverted by $\mathbf{U}^{H}$. The receiver knows the trained $\mathbf{U}_{p}$ and prescribed pilots, and therefore the effective pilots $\mathbf{U}_{p}\mathbf{s}_{p,m}$. Moreover, setting $\mathbf{U}_{p}=\mathbf{I}_{N_p}$ retains the original pilot sequence, while also setting $\mathbf{U}_{d}=\mathbf{I}_{N_d}$ recovers conventional OFDM.

\subsubsection{Communication Processing}

\begin{figure}
	\centering
	\includegraphics[width=0.998\linewidth]{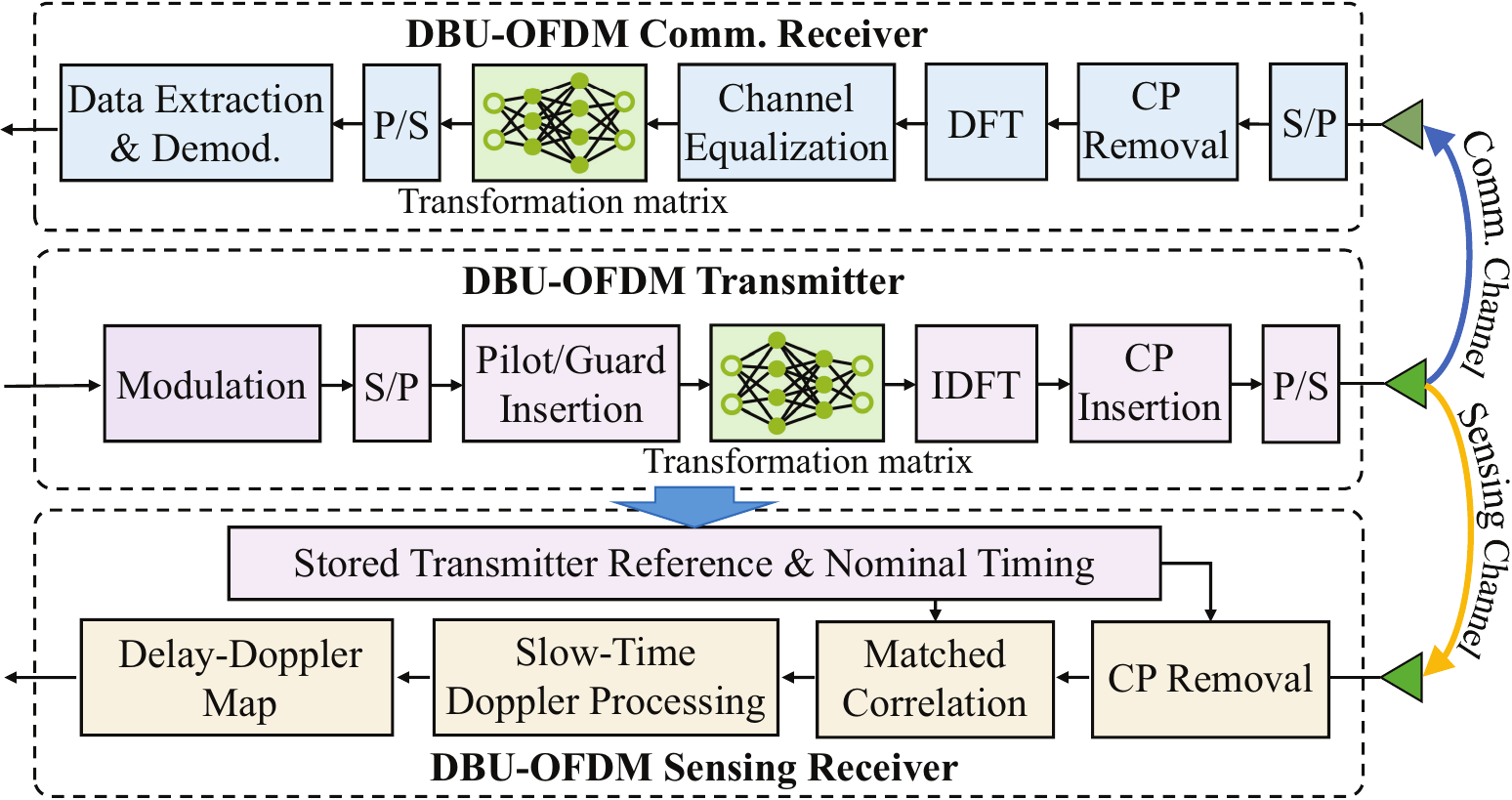}
	    \vspace{-10pt}
    {
	\caption{Communication and sensing processing of the proposed DBU-OFDM waveform. The shared transmitter applies the resource-structured unitary transform before the IDFT. In the communication branch, DFT-based single-tap equalization is followed by the inverse unitary transform. The sensing branch performs matched correlation with the known transmitted waveform under the periodic or aperiodic delay model.}
	\label{fig:dbu_transceiver_pipeline}
    }
\end{figure}

Fig.~\ref{fig:dbu_transceiver_pipeline} illustrates the shared DBU-OFDM transmitter and the corresponding communication and sensing receiver operations. For DBU-OFDM communication with a CP\footnote{For CP-free communication with a zero guard, receiver-side overlap-and-add (OLA) converts the linear channel convolution into an equivalent $N$-point circular convolution, thereby yielding the same DFT-based single-tap equalization structure as the CP-based branch \cite{1175490,4098993}.}, assume that the CP covers the channel delay spread, the channel remains approximately time invariant over one useful OFDM symbol, and the intrasymbol Doppler variation and ICI are negligible. Under these assumptions, the received useful block $\mathbf{y}_{m}\in\mathbb{C}^{N\times 1}$ after CP removal is
\begin{align}
    &\mathbf{y}_{m}
    =\mathbf{H}_{\mathrm{circ},m}\mathbf{F}_{N}^{H}
      \mathbf{U}\mathbf{s}_{m}+\mathbf{w}_{m},                 \label{eq:dbu_received_time}
\end{align}
where $\mathbf{H}_{\mathrm{circ},m}\in\mathbb{C}^{N\times N}$ and $\mathbf{w}_{m}\in\mathbb{C}^{N\times 1}$ denote the corresponding circulant channel matrix and additive noise, respectively. Under the stated CP condition, $\mathbf{H}_{\mathrm{circ},m}$ represents the circular convolution of the channel with the transmitted DBU-OFDM block. Note that applying the unitary transformation before the IDFT does not alter the circulant channel model established by the CP.

Applying the DFT to Eq. \eqref{eq:dbu_received_time} gives
\begin{align}
    &\mathbf{r}_{m}
    =\mathbf{F}_{N}\mathbf{y}_{m}
    =\mathbf{\Lambda}_{m}\mathbf{U}\mathbf{s}_{m}
      +\mathbf{w}_{m}^{(f)},                                  \label{eq:dbu_received_freq}
\end{align}
where $\mathbf{r}_{m}\in\mathbb{C}^{N\times 1}$ denotes the frequency-domain received vector, $\mathbf{\Lambda}_{m}=\mathbf{F}_{N}\mathbf{H}_{\mathrm{circ},m}\mathbf{F}_{N}^{H}\in\mathbb{C}^{N\times N}$ denotes the frequency-domain channel matrix, and $\mathbf{w}_{m}^{(f)}=\mathbf{F}_{N}\mathbf{w}_{m}$ denotes the transformed noise with the same distribution as $\mathbf{w}_{m}$. Since $\mathbf{H}_{\mathrm{circ},m}$ is circulant, we have that $\mathbf{\Lambda}_{m}$ is diagonal. Therefore, channel equalization remains a per-subcarrier operation on the transformed resource vector $\mathbf{U}\mathbf{s}_{m}$, followed by the inverse unitary transformation $\mathbf{U}^H$.

The receiver employs a pilot-aided channel estimator that accounts for the known effective-pilot pattern to estimate the diagonal entries of $\mathbf{\Lambda}_{m}$. Let $\mathbf{G}_{m}\in\mathbb{C}^{N\times N}$ denote the corresponding diagonal single-tap equalizer, and let $\hat{\mathbf{s}}_{m}\in\mathbb{C}^{N\times 1}$ denote the estimate of the original resource vector. The estimated vector is
\begin{align}
    &\hat{\mathbf{s}}_{m}
    =\mathbf{U}^{H}\mathbf{G}_{m}\mathbf{r}_{m}.               \label{eq:dbu_recovery}
\end{align}
Specifically, $\mathbf{G}_{m}$ first estimates the transformed resource vector $\mathbf{U}\mathbf{s}_{m}$, and $\mathbf{U}^{H}$ then recovers the original resource vector. Under ideal noiseless zero-forcing equalization satisfying $\mathbf{G}_{m}\mathbf{\Lambda}_{m}=\mathbf{I}_{N}$, this processing recovers $\mathbf{s}_{m}$ exactly. For a general equalizer, $\hat{\mathbf{s}}_{m}$ is the corresponding symbol estimate. Moreover, the inverse unitary transformation does not amplify the total post-equalization noise energy, whereas any noise enhancement caused by channel inversion is governed by $\mathbf{G}_{m}$.

\begin{Rem}[Communication compatibility]\label{rem:dbu_communication}
The structure design of DBU-OFDM retains the DFT channel diagonalization under the stated CP-OFDM conditions. Moreover, the data-block transformation $\mathbf{U}_{d}$ can spread each data symbol across multiple data subcarriers, thereby enabling a frequency-domain diversity mechanism in frequency-selective channels. The communication behavior and performance obtained when this mechanism is optimized for communication objectives were analyzed in \cite{luo2026dbu}. Accordingly, we focus on optimizing DBU-OFDM for sensing and evaluating its sensing performance.
\end{Rem}

\subsubsection{Sensing Processing}
For sensing, the transmitted DBU-OFDM frame $\bar{\mathbf{X}}$ serves as the known reference waveform at the sensing receiver, as shown in Fig.~\ref{fig:dbu_transceiver_pipeline}. We denote its self-ambiguity function by $\bar{\Gamma}_{\chi}[k,q]$, which is evaluated from $\bar{\mathbf{X}}$ according to Eq. \eqref{eq:self_af} under the P-ACF or A-ACF mode.

Using Eq. \eqref{eq:af_origin} together with the unitarity of $\mathbf{U}$ and $\mathbf{F}_{N}$, we have
\begin{align}
    &\|\bar{\mathbf{X}}\|_{F}^{2}
    =\|\bar{\mathbf{S}}\|_{F}^{2}
    =\|\mathbf{S}\|_{F}^{2}
    =\|\mathbf{X}\|_{F}^{2},                                  \nonumber\\
    &\bar{\Gamma}_{\chi}[0,0]
    =\Gamma_{\chi}[0,0],
    \qquad \chi\in\{\mathrm{P},\mathrm{A}\}.                   \label{eq:dbu_energy}
\end{align}
For every fixed input frame, these equalities establish that the DBU-OFDM transformation preserves the waveform energy and, consequently, the mainlobe reference in the ROI-normalized objective. For a selected correlation mode $\chi\in\{\mathrm{P},\mathrm{A}\}$, the waveform $\bar{\mathbf{X}}$ generated according to Eq. \eqref{eq:dbu_frame} depends on the trainable blocks $\mathbf{U}_{d}$ and $\mathbf{U}_{p}$. Therefore, the ROI-oriented waveform design is formulated as
\begin{subequations}\label{eq:dbu_roi_optimization}
\begin{align}
    (\mathrm{P1}):\quad
    &\underset{\mathbf{U}_{d},\mathbf{U}_{p}}{\operatorname{max}}
      \quad \mathrm{PSLR}_{\mathrm{ROI}}^{(\chi)}(\bar{\mathbf{X}})
      \tag{\theparentequation}\label{eq:dbu_roi_objective}\\
    &\operatorname{s.t.}\quad
      \mathbf{U}_{\vartheta}^{H}\mathbf{U}_{\vartheta}
      =\mathbf{I}_{N_{\vartheta}},\quad \vartheta\in\{d,p\}.
      \tag{\theparentequation a}\label{eq:dbu_roi_unitarity}
\end{align}
\end{subequations}
Since $\mathbf{U}_{0}=\mathbf{I}_{N_0}$ and $\mathbf{\Pi}$ is orthogonal, the block constraint in Eq. \eqref{eq:dbu_roi_unitarity} is equivalent to the full-size unitary condition $\mathbf{U}^{H}\mathbf{U}=\mathbf{I}_{N}$.

During training, we use the differentiable loss in Eq. \eqref{eq:roi_loss} as a surrogate for problem (P1). However, a soft penalty proportional to $\sum_{\vartheta\in\{d,p\}}\|\mathbf{U}_{\vartheta}^{H}\mathbf{U}_{\vartheta}-\mathbf{I}_{N_{\vartheta}}\|_{F}^{2}$ with a finite weight does not require the unitary residual to vanish. We therefore adopt a constraint-preserving parameterization to ensure exact unitarity throughout optimization.

\subsubsection{Householder-Based Trainable Parameterization}

Specifically, we parameterize both resource blocks using complex Householder reflections and diagonal phase rotations.

Let $N_{\vartheta}$ and $K_{\vartheta}$ denote the dimension and Householder depth budget for $\mathbf{U}_{\vartheta}$, respectively, where $K_{\vartheta}$ is a nonnegative integer. For a nonzero trainable vector $\mathbf{v}_{\vartheta,i}\in\mathbb{C}^{N_{\vartheta}\times 1}$, the $i$-th Householder matrix is
\begin{align}
    &\mathbf{H}_{\vartheta,i}
    =\mathbf{I}_{N_{\vartheta}}
     -2\frac{\mathbf{v}_{\vartheta,i}\mathbf{v}_{\vartheta,i}^{H}}
              {\mathbf{v}_{\vartheta,i}^{H}\mathbf{v}_{\vartheta,i}},
    \qquad i=1,\ldots,K_{\vartheta}.                               \label{eq:householder_matrix}
\end{align}
Specifically, the normalized rank-one matrix $\mathbf{v}_{\vartheta,i}\mathbf{v}_{\vartheta,i}^{H}/(\mathbf{v}_{\vartheta,i}^{H}\mathbf{v}_{\vartheta,i})$ is a Hermitian idempotent projector. Therefore, each $\mathbf{H}_{\vartheta,i}$ is Hermitian and unitary. Moreover, we introduce the trainable diagonal phase matrix
\begin{align}
    &\mathbf{D}_{\vartheta}
    =\operatorname{diag}\!\left(
      e^{j\theta_{\vartheta,1}},\ldots,e^{j\theta_{\vartheta,N_{\vartheta}}}
      \right),                                                  \label{eq:diagonal_phase}
\end{align}
where $\theta_{\vartheta,\ell}\in\mathbb{R}$ denotes a trainable phase for $\ell=1,\ldots,N_{\vartheta}$.
The unitary matrix $\mathbf{U}_{\vartheta}$ is then constructed as
\begin{align}
    &\mathbf{U}_{\vartheta}
    =\mathbf{D}_{\vartheta}
     \mathbf{H}_{\vartheta,K_{\vartheta}}\cdots
     \mathbf{H}_{\vartheta,1}.                                    \label{eq:householder_parameterization}
\end{align}
When $K_{\vartheta}=0$, the reflection product is defined as $\mathbf{I}_{N_{\vartheta}}$, and thus $\mathbf{U}_{\vartheta}=\mathbf{D}_{\vartheta}$. The same construction is used for $\mathbf{U}_{d}$ and $\mathbf{U}_{p}$, whereas the null block remains $\mathbf{I}_{N_0}$. Since every factor is unitary, the resulting blocks satisfy the unitary constraints without projection or re-orthogonalization.

To determine whether the Householder construction restricts the feasible unitary transformations, we characterize the depth required for complete unitary representation.
\begin{Lemma}[Representation capability of the Householder parameterization]\label{lemma:1}
For $\vartheta\in\{d,p\}$, a Householder depth budget satisfying $K_{\vartheta}\geq N_{\vartheta}-1$ is sufficient for Eq. \eqref{eq:householder_parameterization} to represent any $N_{\vartheta}\times N_{\vartheta}$ unitary matrix exactly.
\end{Lemma}
\begin{IEEEproof}
    Please refer to Appendix \ref{app:A} for detailed proof.
\end{IEEEproof}

\begin{Rem}[Depth requirements for the resource blocks]\label{rem:householder_depth}
Applying Lemma~\ref{lemma:1} to the data and pilot blocks shows that $K_d\geq N_d-1$ and $K_p\geq N_p-1$ are sufficient to represent arbitrary $\mathbf{U}_{d}$ and $\mathbf{U}_{p}$ exactly. When $K_{\vartheta}<N_{\vartheta}-1$, every constructed matrix remains exactly unitary, but complete representation of the $N_{\vartheta}\times N_{\vartheta}$ unitary group is no longer guaranteed. Moreover, note that Lemma~\ref{lemma:1} characterizes the representation capability of the parameterization but does not guarantee that gradient-based training converges to a globally optimal unitary matrix.
\end{Rem}

During training, the reflection vectors and diagonal phases are optimized by minimizing $\mathcal{L}_{\chi}(\bar{\mathbf{X}})$ in Eq. \eqref{eq:roi_loss}. 

\subsection{Sensing Subcarrier Allocation}\label{sec:III-C}

\subsubsection{Sensing and Communication Resource Partition}

The data subspace originally carries random communication symbols whose values are determined by the transmitted payload. To increase the controllable waveform-shaping degrees of freedom, we reserve a subset of the data subcarriers for known and trainable sensing symbols. 



Specifically, let $\bm{\mathcal{K}}_s$ denote the sensing-subcarrier set selected from $\bm{\mathcal{K}}_d$, while $\bm{\mathcal{K}}_c$ contains the remaining communication subcarriers. These two sets satisfy
\begin{align}
    &\bm{\mathcal{K}}_d
      =\bm{\mathcal{K}}_s\cup\bm{\mathcal{K}}_c,
      \qquad
      \bm{\mathcal{K}}_s\cap\bm{\mathcal{K}}_c=\varnothing,\nonumber\\
    &|\bm{\mathcal{K}}_s|=K_s,
      \qquad
      |\bm{\mathcal{K}}_c|=K_c=N_d-K_s.                       \label{eq:sensing_comm_partition}
\end{align}
Increasing $K_s$ will provide more controllable sensing resources and reduce the number of subcarriers available for the communication payload. Therefore, $K_s$ directly determines the sensing-resource budget and the corresponding sensing-communication resource tradeoff.

Let $\mathbf{E}_s\in\{0,1\}^{N_d\times K_s}$ and $\mathbf{E}_c\in\{0,1\}^{N_d\times K_c}$ denote the binary matrices that embed the sensing and communication blocks into their corresponding positions within the data subspace. Moreover, let $\mathbf{S}_s\in\mathbb{C}^{K_s\times M}$ and $\mathbf{S}_c\in\mathbb{C}^{K_c\times M}$ collect the sensing and communication symbols, and let $\mathbf{U}_s$ and $\mathbf{U}_c$ denote their unitary transformations. The partitioned data block and its transformation are expressed as
\begin{align}
    &\mathbf{S}_d
      =\mathbf{E}_s\mathbf{S}_s
       +\mathbf{E}_c\mathbf{S}_c,\nonumber\\
    &\mathbf{U}_d
      =\mathbf{E}_s\mathbf{U}_s\mathbf{E}_s^{T}
       +\mathbf{E}_c\mathbf{U}_c\mathbf{E}_c^{T},\nonumber\\
    &\mathbf{U}_d\mathbf{S}_d
      =\mathbf{E}_s\mathbf{U}_s\mathbf{S}_s
       +\mathbf{E}_c\mathbf{U}_c\mathbf{S}_c.                  \label{eq:data_unitary_action}
\end{align}
Since the two embedding matrices select disjoint supports, the sensing and communication symbols do not mix during the linear frequency-domain transformation. The matrix $\mathbf{S}_s$ contains known and trainable sensing symbols normalized to unit average power, whereas $\mathbf{S}_c$ contains random communication symbols.

\subsubsection{Array-Factor-Based Support Initialization}

To characterize the role of the sensing support, we separate the transmitted DBU-OFDM frame into the sensing/pilot component $\bar{\mathbf{X}}_{sp}$ and the random communication component $\bar{\mathbf{X}}_c$, i.e., $\bar{\mathbf{X}} = \bar{\mathbf{X}}_{sp} + \bar{\mathbf{X}}_c$. The null resources do not contribute to either component. Let $\Gamma_{\chi}(\mathbf{A},\mathbf{B})[k,q]$ denote the cross-ambiguity term between $\mathbf{A}$ and $\mathbf{B}$. The complete self-ambiguity function is then decomposed as
\begin{align}
    &\bar{\Gamma}_{\chi}[k,q]
      =\Gamma_{\chi}(\bar{\mathbf{X}}_{sp},
       \bar{\mathbf{X}}_{sp})[k,q]
       +\Gamma_{\chi}(\bar{\mathbf{X}}_c,
       \bar{\mathbf{X}}_c)[k,q]\nonumber\\
    &\qquad
       +\Gamma_{\chi}(\bar{\mathbf{X}}_{sp},
       \bar{\mathbf{X}}_c)[k,q]
       +\Gamma_{\chi}(\bar{\mathbf{X}}_c,
       \bar{\mathbf{X}}_{sp})[k,q].                            \label{eq:ambiguity_component_decomposition}
\end{align}
The first two terms describe the sensing/pilot and communication self-ambiguity contributions, respectively, while the last two are their cross-ambiguity terms. Their coherent superposition determines the complete ROI sidelobes.

During initialization, $\bar{\mathbf{X}}_c$ varies with the random payload, whereas the sensing excitation and pilots assigned to each candidate support are fixed. To obtain a data-independent starting support, we therefore optimize a proxy of the sensing/pilot self-ambiguity term.

Let $\bm{\mathcal{A}}=\bm{\mathcal{K}}_s\cup\bm{\mathcal{K}}_p$ denote the combined support. We assign unit excitation to these resources and consider their zero-Doppler periodic delay response. Since this response is proportional to the inverse DFT of the support indicator, its magnitude is proportional to that of the following array factor
\begin{align}
    &A_{\bm{\mathcal{A}}}[k]
      =\sum_{\ell\in\bm{\mathcal{A}}}
       e^{-j2\pi \ell k/N},\nonumber\\
    &|A_{\bm{\mathcal{A}}}[k]|^2
      =\sum_{\ell_1\in\bm{\mathcal{A}}}
       \sum_{\ell_2\in\bm{\mathcal{A}}}
       e^{-j2\pi(\ell_1-\ell_2)k/N}.                          \label{eq:support_arrayfactor}
\end{align}
The first row represents a support-only approximation of the sensing/pilot delay sidelobe, while the second row shows its dependence on the pairwise differences between the occupied subcarrier indices. Thus, the sensing support directly controls the effective frequency aperture and difference-coarray structure of this approximation.

Let $\Omega_k\subseteq\mathbb{Z}_N\setminus\{0\}$ denote the delay-axis projection of $\Omega_{\mathrm{SL}}$. For a fixed sensing-resource budget $K_s$, the support initialization is formulated as
\begin{subequations}\label{eq:arrayfactor_support_problem}
\begin{align}
    (\mathrm{P2}):\quad
    &\underset{\bm{\mathcal{K}}_s}{\operatorname{min}}
      \max_{k\in\Omega_k}\quad
      \frac{
       |A_{\bm{\mathcal{A}}}[k]|
      }{K_s+N_p}
      \tag{\theparentequation}\label{eq:arrayfactor_support_objective}\\
    &\operatorname{s.t.}\quad
      \bm{\mathcal{K}}_s\subseteq\bm{\mathcal{K}}_d
      \tag{\theparentequation a}\label{eq:arrayfactor_support_constraint}
\end{align}
\end{subequations}
The denominator is the zero-delay array-factor magnitude and remains constant for a prescribed $K_s$. Moreover, the pilots are included in every candidate support, while the null subcarriers are excluded from the candidate set. We approximately solve problem (P2) through sequential support construction followed by single-subcarrier exchange refinement.


The array-factor problem optimizes only the sensing/pilot term in Eq. \eqref{eq:ambiguity_component_decomposition} under the stated support-only assumptions. The resulting support is therefore not guaranteed to minimize the complete objective after the communication and cross-ambiguity terms are included. Moreover, the proxy does not account for the actual symbol values, unitary transformations, Doppler structure, or A-ACF boundaries. These factors are included in the subsequent joint optimization through the complete self-ambiguity loss.

Accordingly, the joint DBU-OFDM waveform design is formulated as
\begin{subequations}\label{eq:joint_sensing_allocation_problem}
\begin{align}
    (\mathrm{P2.1}):\quad
    &\underset{
       \substack{\bm{\mathcal{K}}_s,\mathbf{S}_s,\\
       \mathbf{U}_s,\mathbf{U}_c,\mathbf{U}_p}
      }{\operatorname{min}}
      \quad \mathcal{L}_{\chi}(\bar{\mathbf{X}})
      \tag{\theparentequation}\label{eq:joint_sensing_allocation_objective}\\
    &\operatorname{s.t.}\quad
      \bm{\mathcal{K}}_s\subseteq\bm{\mathcal{K}}_d
      \tag{\theparentequation a}\label{eq:joint_sensing_support_constraint}\\
    &\hphantom{\operatorname{s.t.}\quad}
      |\bm{\mathcal{K}}_s|=K_s
      \tag{\theparentequation b}\label{eq:joint_sensing_cardinality_constraint}\\
    &\hphantom{\operatorname{s.t.}\quad}
      \|\mathbf{S}_s\|_F^2=K_sM
      \tag{\theparentequation c}\label{eq:joint_sensing_power_constraint}\\
    &\hphantom{\operatorname{s.t.}\quad}
      \mathbf{U}_\vartheta^{H}\mathbf{U}_\vartheta
      =\mathbf{I}_{K_\vartheta},\quad \vartheta\in\{s,c,p\}.
      \tag{\theparentequation d}\label{eq:joint_sensing_unitarity_constraint}
\end{align}
\end{subequations}

The array-factor solution initializes the hard support. When differentiable support refinement is enabled, the support is refined through a straight-through top-$K_s$ estimator using gradients of the complete loss, while $\mathbf{S}_s$, $\mathbf{U}_s$, $\mathbf{U}_c$, and $\mathbf{U}_p$ are jointly optimized.

\section{Numerical Results}\label{sec:IV}

In this section, we present numerical results to evaluate the proposed ROI-oriented DBU-OFDM waveform under the P-ACF and A-ACF modes, with emphasis on ROI sidelobe suppression, sensing-subcarrier allocation, and the sensing-resource tradeoff. Unless otherwise specified, the common discrete-domain waveform, ROI, resource, and training parameters are listed in Table~\ref{tab:numerical_parameters}.

\begin{table}[!t]
    \centering
    \caption{Common discrete-domain waveform, ROI, resource, and training parameters.}
    \label{tab:numerical_parameters}
    \footnotesize
    \begin{tabular}{@{}lcc@{}}
        \toprule
        \textbf{Parameter} & \textbf{Notation} & \textbf{Value} \\
        \midrule
        Total subcarriers & $N$ & 256 \\
        OFDM symbols per frame & $M$ & 128 \\
        Delay indices in the ROI & $k$ & $0,\ldots,31$ \\
        Doppler indices in the ROI & $q$ & $-5,\ldots,5$ \\
        Edge guards & $N_g$ & 4 per edge \\
        DC nulls & $N_{dc}$ & 2 \\
        Comb-type pilots & $N_p$ & 16 \\
        Data subcarriers & $N_d$ & 230 \\
        Communication modulation & -- & 64-QAM \\
        Independent frames per batch & -- & 100 \\
        Smooth-maximum parameter & $\alpha$ & 5 \\
        Optimizer & -- & AdamW \\
        Learning rate & -- & $10^{-2}$ \\
        Maximum training epochs & -- & $1.2\times10^{4}$ \\
        Householder depth & $K_{\vartheta}$ & $N_{\vartheta}-1$ \\
        \bottomrule
    \end{tabular}
\end{table}

To relate these discrete parameters to a representative physical sensing scenario, we consider $f_c=28$ GHz, $B=30$ MHz, and $N_{\mathrm{CP}}=N_{\mathrm{ZG}}=64$. The sampling bandwidth gives a range-bin spacing of approximately 5 m, and the delay indices $k=0,\ldots,31$ therefore represent the range interval $[0,160)$ m. The representative guard lengths give $T_{\mathrm{P}}=T_{\mathrm{A}}=10.67~\mu$s. Thus, the Doppler indices $q=-5,\ldots,5$ have a Doppler-bin spacing of 732.4 Hz and represent the velocity interval $[-19.6,19.6]$ m/s. Moreover, $N_{\mathrm{CP}}=N_{\mathrm{ZG}}=64>31$ satisfies the delay conditions of the P-ACF and A-ACF models. The largest absolute Doppler index $q_{\max}^{(\chi)}=5$ gives $\phi_{\mathrm{intra},\max}^{(\chi)}=0.196$ rad $<0.2$ rad in Eq. \eqref{eq:small_doppler}, which supports neglecting the intrasymbol Doppler variation and the resulting ICI under this representative configuration.

The communication symbols are independently drawn from a 64-QAM constellation and normalized to unit average power. The input comb-type pilot symbols are fixed as $P_{n,m}=(1+j)/\sqrt{2}$ and therefore also have unit power. Moreover, the sensing symbols are normalized to unit average power.

For each trainable resource block $\mathbf{U}_{\vartheta}$, where $\vartheta\in\{s,c,p\}$, the number of Householder reflections is set to $K_{\vartheta}=N_{\vartheta}-1, \forall \vartheta\in\{s,c,p\}$. The null block remains $\mathbf{U}_0=\mathbf{I}_{N_0}$ and does not contain trainable reflections.

\subsection{Validation Under a Fully Occupied OFDM Resource Structure}\label{sec:IV-A}

We first validate the proposed waveform model and training method under a fully occupied OFDM structure, which enables comparison with the known delay-domain properties of conventional OFDM. In this section, all $N=256$ subcarriers carry independent unit-power 64-QAM symbols. Thus, $N_d=N$, $N_p=N_0=K_s=0$, and only $\mathbf{U}_d$ is trained separately for the two correlation modes. We set $\Omega_{\mathrm{ROI}}=\mathbb{Z}_N\times\{0\}$ and $\Omega_{\mathrm{SL}}=(\mathbb{Z}_N\setminus\{0\})\times\{0\}$ and minimize the loss in Eq. \eqref{eq:roi_loss} over this zero-Doppler delay slice. 

\begin{figure*}[!t]
    \centering
    \subfloat[Common initial $\left|\mathbf{U}_d\right|$.]{%
        \includegraphics[width=0.34\textwidth, trim={34bp 5bp 48bp 10bp},
    clip]{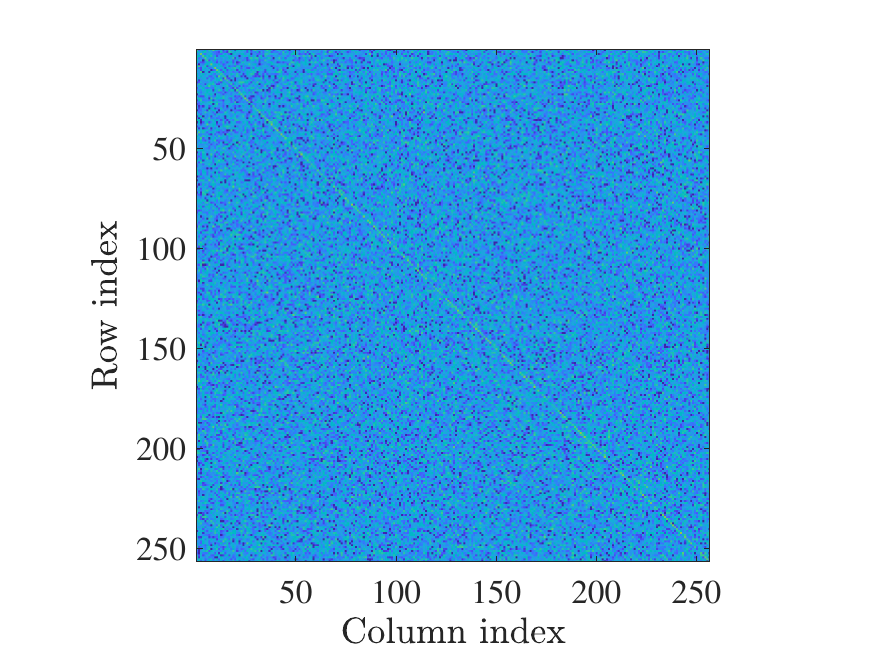}%
        \label{fig:common_unitary_init}}
    \hspace{-0.03\textwidth}
    \subfloat[Trained $\left|\mathbf{U}_d\right|$ under P-ACF.]{%
        \includegraphics[width=0.34\textwidth, trim={34bp 5bp 48bp 10bp},
    clip]{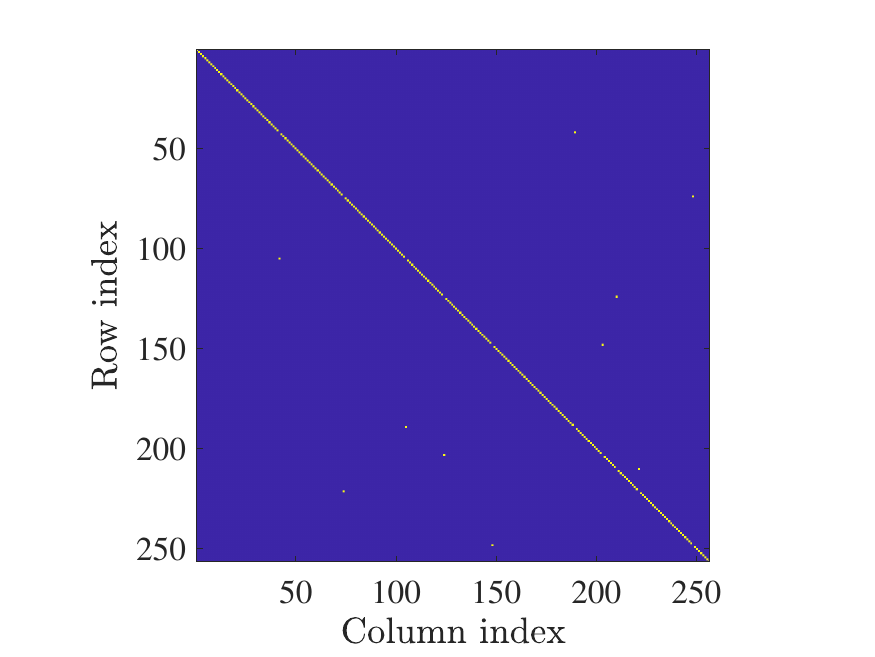}%
        \label{fig:pacf_unitary_trained}}
    \hspace{-0.025\textwidth}
    \subfloat[Trained $\left|\mathbf{U}_d\right|$ under A-ACF.]{%
        \includegraphics[width=0.36\textwidth, trim={14bp 5bp 48bp 10bp},
    clip]{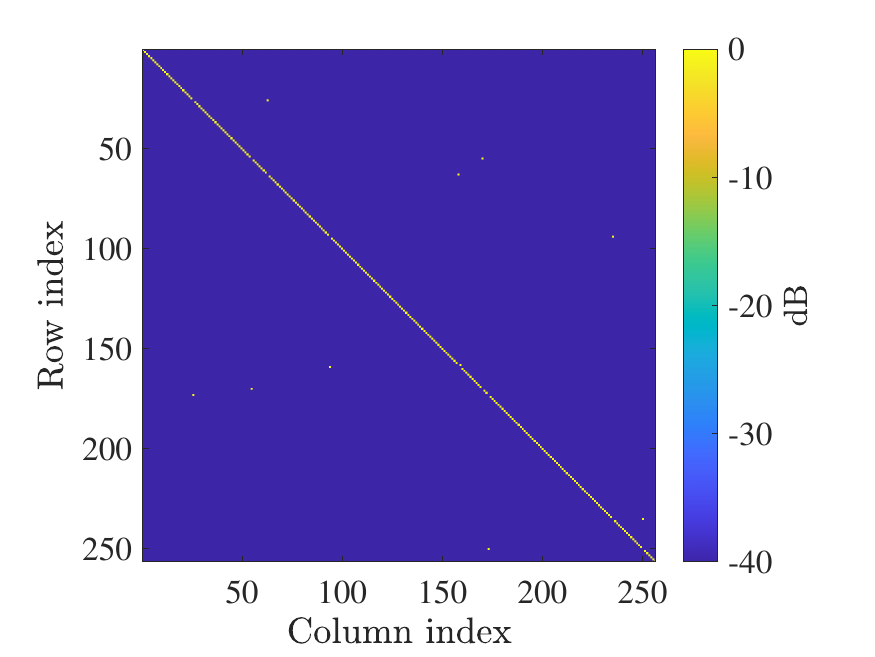}%
        \label{fig:aacf_unitary_trained}}
    \caption{Data-block unitary matrices before and after training under the fully occupied OFDM resource structure. The P-ACF and A-ACF experiments use the same dense initialization. The entry magnitudes are shown in dB.}
    \label{fig:full_resource_unitary_evolution}
\end{figure*}

\begin{figure}[!t]
    \centering
    \subfloat[Rows and columns 1-16.]{%
        \includegraphics[width=0.202\textwidth, trim={44bp 5bp 84bp 10bp},
    clip]{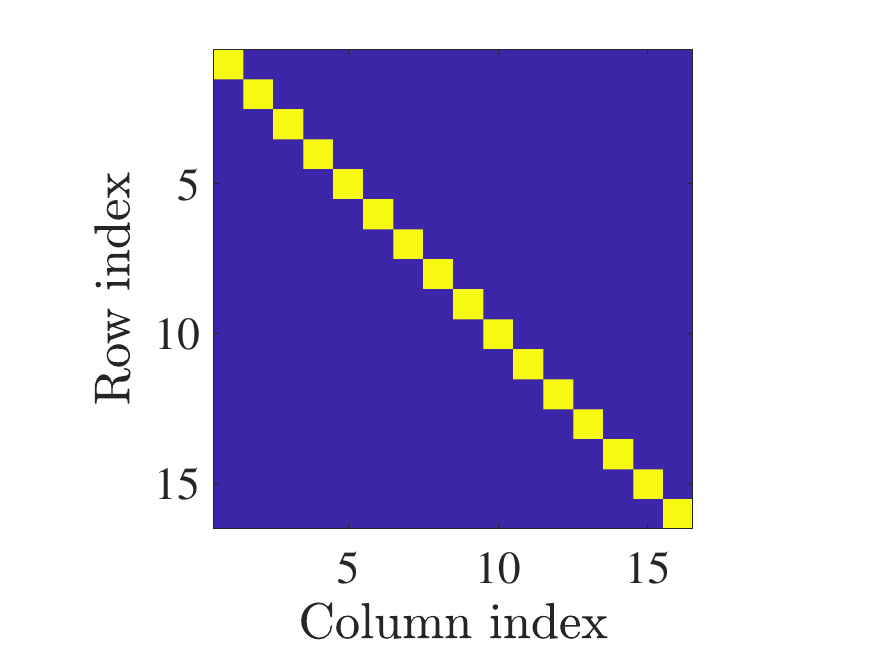}%
        \label{fig:pacf_unitary_trained_subplot1}}
   \hspace{0.01\textwidth}
    \subfloat[Rows and columns 208-223.]{%
        \includegraphics[width=0.248\textwidth, trim={14bp 5bp 48bp 10bp},
    clip]{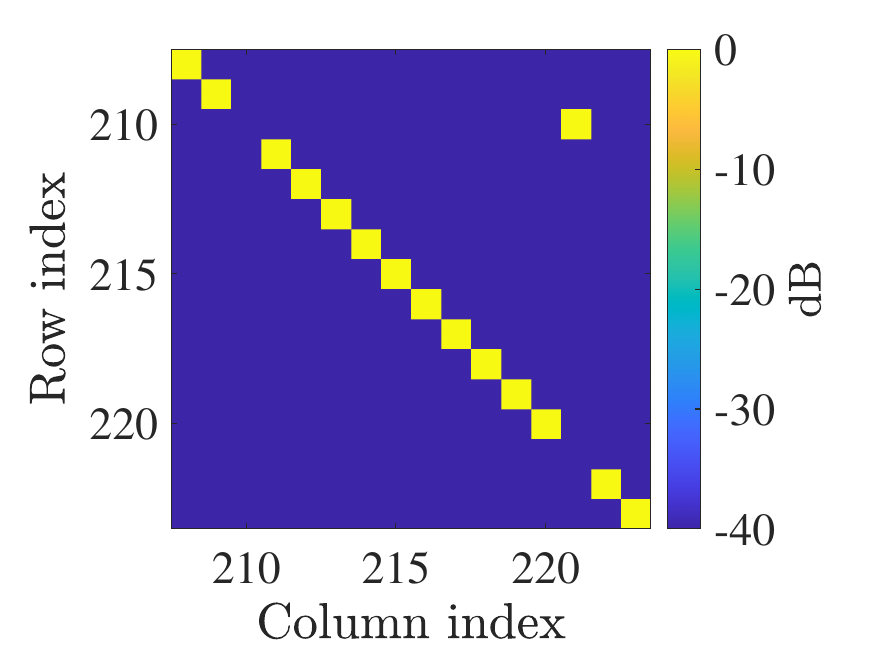}%
        \label{fig:pacf_unitary_trained_subplot2}}
    \caption{Enlarged views of two principal submatrices extracted from the trained P-ACF data-block unitary matrix in Fig.~\ref{fig:pacf_unitary_trained}.}
    \label{fig:full_resource_unitary_evolution_subplot}
\end{figure}

Fig.~\ref{fig:full_resource_unitary_evolution} shows that training converts the common dense initialization into matrices with one dominant, near-$0$-dB entry per row and column under both correlation modes. Most dominant entries lie on the main diagonal. The enlarged P-ACF submatrices in Fig.~\ref{fig:full_resource_unitary_evolution_subplot} reveal a few off-diagonal entries and therefore a small symbol reordering. We quantify this near-monomial structure using the monomiality ratio
\begin{align}
    &\rho_{\mathrm{mon}}(\mathbf{U}_d)
    =\frac{1}{N}\max_{\pi\in\mathfrak{S}_N}
      \sum_{n=0}^{N-1}\left|[\mathbf{U}_d]_{n,\pi(n)}\right|^2,
                                                               \label{eq:monomiality_ratio}
\end{align}
where $\mathfrak{S}_N$ denotes the set of all permutations of $\mathbb{Z}_N$. A monomial unitary matrix has $\rho_{\mathrm{mon}}=1$. The common initialization gives $\rho_{\mathrm{mon}}=0.0268$, which increases to $0.9980$ under P-ACF training and $0.9989$ under A-ACF training.

These values indicate that the trained data-block transformations approach the monomial unitary form
\begin{align}
    &\mathbf{U}_d^{(\chi)}
    \approx\mathbf{D}_{\chi}\mathbf{P}_{\chi},
    \quad \chi\in\{\mathrm{P},\mathrm{A}\},\nonumber\\
    &\mathbf{D}_{\chi}
    =\operatorname{diag}\!\left(
      e^{j\phi_{\chi,0}},\ldots,e^{j\phi_{\chi,N-1}}
      \right),
    \quad
    \mathbf{P}_{\chi}^{T}\mathbf{P}_{\chi}=\mathbf{I}_N,
                                                               \label{eq:trained_monomial_form}
\end{align}
where $\mathbf{D}_{\chi}$ is a diagonal phase matrix and $\mathbf{P}_{\chi}\in\{0,1\}^{N\times N}$ is a permutation matrix. Accordingly, the learned transformation approximately preserves a one-to-one mapping between the data symbols and subcarriers. It mainly permutes the data symbols and applies per-subcarrier phase rotations, with negligible mixing across subcarriers.

Under the P-ACF mode, convergence to this monomial structure agrees with the established result that conventional OFDM minimizes the expected sidelobe power at every nonzero delay under QAM/PSK signaling \cite{liu2025cp}. For the A-ACF mode, the corresponding analysis establishes only local optimality under the expected integrated sidelobe power criterion \cite{liu2025cp}. Nevertheless, the trained matrix converges to the same OFDM-equivalent monomial family observed under the P-ACF mode. Although this convergence does not prove global optimality, it provides numerical evidence supporting the conjecture that an OFDM-equivalent monomial transformation may attain the global optimum of the considered zero-Doppler A-ACF delay-slice PSLR problem within the unitary waveform class.

\begin{figure}[!t]
    \centering
    \subfloat[P-ACF.]{%
        \includegraphics[width=0.99\linewidth]{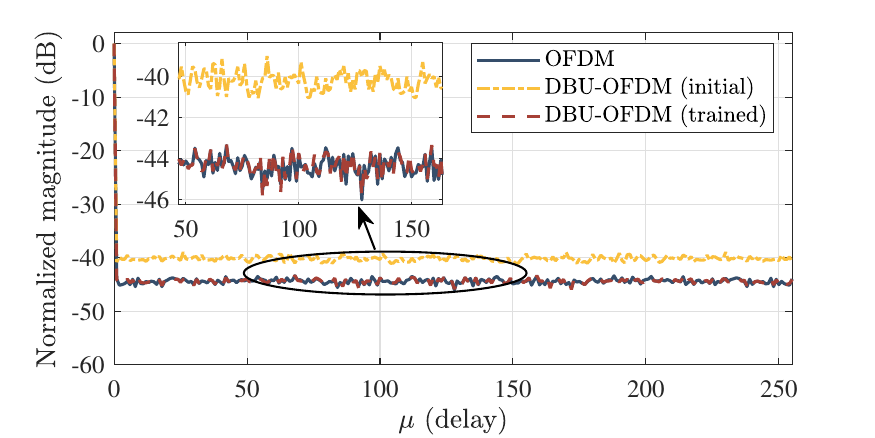}%
        \label{fig:full_resource_pacf_delay_cut}}
    \\
    \vspace{-10pt}
    \subfloat[A-ACF.]{%
        \includegraphics[width=0.99\linewidth]{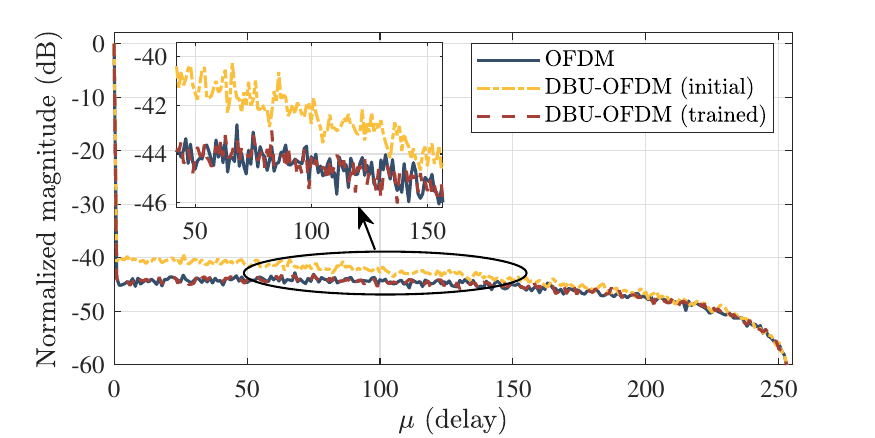}%
        \label{fig:full_resource_aacf_delay_cut}}
    \caption{Normalized zero-Doppler delay cuts under the fully occupied OFDM resource structure. The noiseless OFDM, initialized DBU-OFDM, and trained DBU-OFDM PSLRs are $43.37$, $38.96$, and $43.30$ dB in the P-ACF mode, respectively, and $42.80$, $39.35$, and $43.04$ dB in the A-ACF mode, respectively.}
    \label{fig:full_resource_delay_cuts}
\end{figure}

The zero-Doppler delay cuts in Fig.~\ref{fig:full_resource_delay_cuts} show that the trained DBU-OFDM and conventional OFDM curves nearly overlap over the nonzero delay indices in both correlation modes. The small residual differences are mainly attributed to the deviation of the trained matrices from the exact monomial form in Eq. \eqref{eq:trained_monomial_form}, as indicated by $\rho_{\mathrm{mon}}(\mathbf{U}_d)=0.9980<1$, and the finite-sample variation associated with averaging over random QAM frames. In contrast, the common dense initialization produces substantially higher delay sidelobes. Training improves the PSLR by $4.34$ dB in the P-ACF mode and $3.69$ dB in the A-ACF mode. These improvements demonstrate the effectiveness of the DBU-OFDM training procedure in suppressing the delay sidelobes under both correlation modes.

\subsection{Pilot-Only Ambiguity Shaping}\label{sec:IV-B}

We next consider the practical resource structure in Table~\ref{tab:numerical_parameters} and set $K_s=0$, so that no subcarrier is reserved for sensing. The communication subcarriers continue to carry random 64-QAM symbols, while the prescribed pilot resources provide an additional waveform-shaping opportunity. To isolate this opportunity, we set $\mathbf{U}_d=\mathbf{U}_c=\mathbf{I}_{N_d}$ and optimize only the pilot resources. We compare two pilot designs. The phase-only (PO) design optimizes the phase of every pilot symbol while maintaining unit amplitude, whereas the unitary (UN) design applies the trainable pilot block $\mathbf{U}_p$ to the prescribed pilot vector. Under the UN design, the effective pilots may have different amplitudes and phases, although their total energy is preserved by $\mathbf{U}_p$.\footnote{The unitary pilot design preserves the total pilot energy but does not impose a lower bound on each effective-pilot amplitude. This may affect the practical channel estimation and the resulting communication-sensing tradeoff.}
\begin{figure*}[!t]
    \centering
    \subfloat[P-ACF at $q=0$.]{%
        \includegraphics[width=0.49\textwidth]{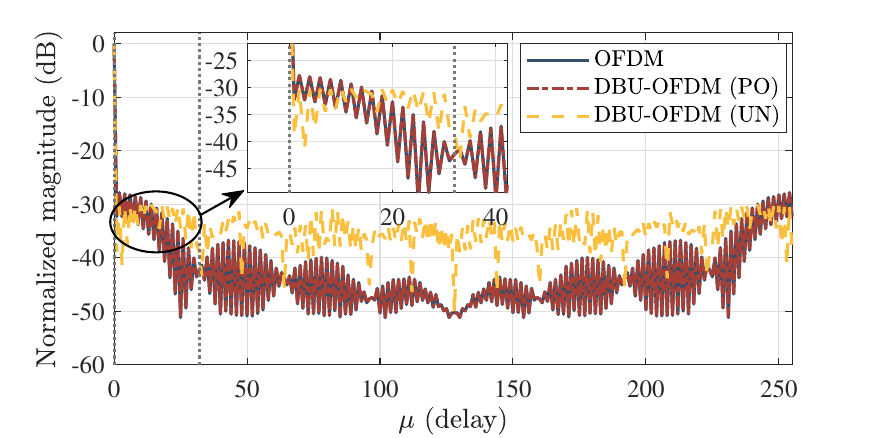}%
        \label{fig:pilot_design_pacf_delay_cut}}
    \hfill
    \subfloat[P-ACF at $k=0$.]{%
        \includegraphics[width=0.49\textwidth]{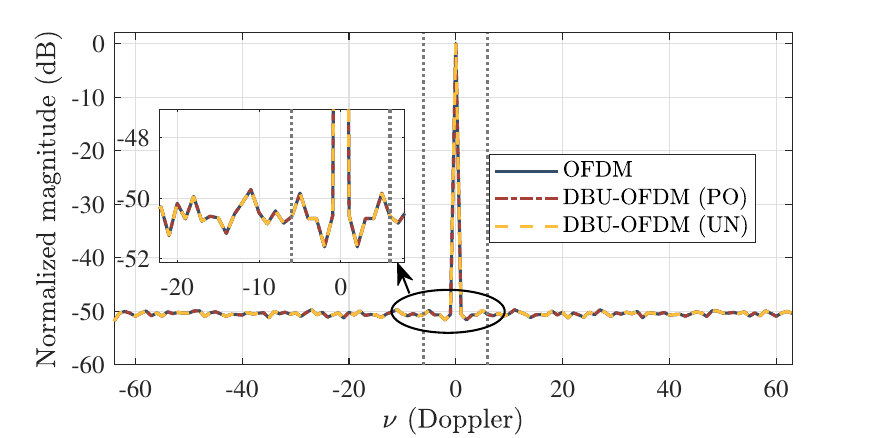}%
        \label{fig:pilot_design_pacf_doppler_cut}}
    \\
    \vspace{-10pt}
    \subfloat[A-ACF at $q=0$.]{%
        \includegraphics[width=0.49\textwidth]{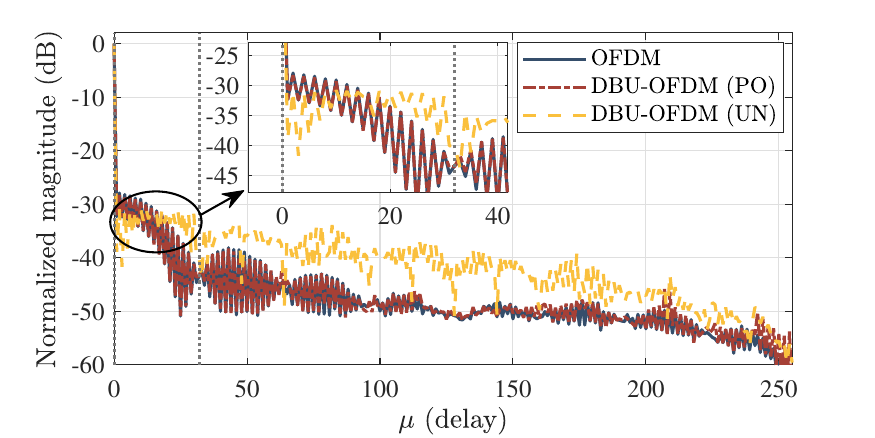}%
        \label{fig:pilot_design_aacf_delay_cut}}
    \hfill
    \subfloat[A-ACF at $k=0$.]{%
        \includegraphics[width=0.49\textwidth]{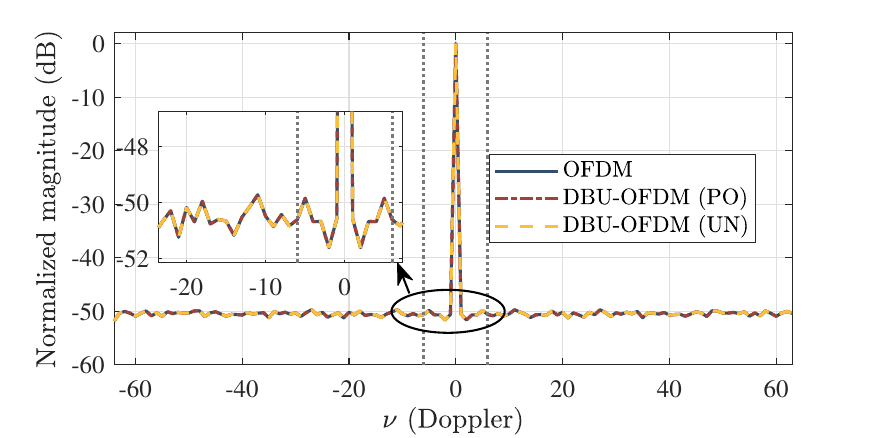}%
        \label{fig:pilot_design_aacf_doppler_cut}}
    \caption{Normalized zero-Doppler delay cuts and zero-delay Doppler cuts for conventional OFDM and DBU-OFDM with phase-only (PO) and unitary (UN) pilot designs. (a) and (b) show the P-ACF results, while (c) and (d) show the A-ACF results.}
    \label{fig:pilot_design_delay_doppler_cuts}
\end{figure*}

Fig.~\ref{fig:pilot_design_delay_doppler_cuts} compares the two pilot designs through the zero-Doppler delay cuts and zero-delay Doppler cuts. Under the P-ACF mode, the phase-only design retains the conventional OFDM delay/Doppler response, whereas the unitary design increases the ROI PSLR from $27.83$ to $30.14$ dB, corresponding to an improvement of approximately $2.31$ dB. Under the A-ACF mode, the phase-only design provides a limited improvement from $27.91$ to $28.27$ dB, while the unitary design increases the ROI PSLR to $30.77$ dB and achieves a gain of approximately $2.86$ dB. Thus, phase adjustment alone provides only a limited A-ACF gain and does not improve the P-ACF, whereas the additional amplitude-redistribution freedom of the unitary pilot block improves the ROI PSLR under both correlation modes.

The absence of a P-ACF gain from a strictly phase-only pilot design follows from the circular correlation property.

\begin{Lemma}[P-ACF invariance under phase-only pilot design]\label{lemma:lemma2}
For a fixed OFDM frame, adjusting only the phases of the pilot symbols will not change the P-ACF sidelobes.
\end{Lemma}

\begin{IEEEproof}
    Please refer to Appendix \ref{app:B} for detailed proof.
\end{IEEEproof}

Lemma~\ref{lemma:lemma2} explains why optimizing only the pilot phases does not provide a P-ACF gain. Note that the same conclusion could also be applied to phase-only adjustment of the sensing subcarriers. However, the A-ACF does not have this phase-invariance property since its finite-overlap correlation contains cross-subcarrier terms. Pilot phases can therefore modify its nonzero-delay sidelobes.

Moreover, the three zero-delay Doppler cuts in Figs.~\ref{fig:pilot_design_pacf_doppler_cut} and \ref{fig:pilot_design_aacf_doppler_cut} overlap exactly since the energy of each symbol is preserved.

\begin{Lemma}[Zero-delay Doppler-cut invariance of DBU-OFDM]\label{lemma:lemma3}
For a fixed input frame, let $\Gamma_{\chi}[k,q]$ and $\bar{\Gamma}_{\chi}[k,q]$ denote the self-ambiguity functions of conventional OFDM and DBU-OFDM with an arbitrary unitary transformation $\mathbf{U}$, respectively. Their zero-delay Doppler cuts are identical under both correlation modes, and thus we have
\begin{align}
    &\bar{\Gamma}_{\chi}[0,q]
    =\Gamma_{\chi}[0,q],
    \quad q\in\mathbb{Z}_M^{\mathrm{c}},\quad
    \chi\in\{\mathrm{P},\mathrm{A}\}.                     \label{eq:zero_delay_invariance}
\end{align}
\end{Lemma}

\begin{IEEEproof}
    Please refer to Appendix \ref{app:C} for detailed proof.
\end{IEEEproof}

The PO and UN pilot designs are both special cases of the unitary DBU-OFDM transformation in Lemma~\ref{lemma:lemma3}. Therefore, their zero-delay Doppler cuts coincide exactly with that of conventional OFDM in Fig.~\ref{fig:pilot_design_pacf_doppler_cut} and Fig.~\ref{fig:pilot_design_aacf_doppler_cut}. At a nonzero delay, unitarity alone does not preserve the delay-dependent correlations. Lemma~\ref{lemma:lemma2} provides a special exception by guaranteeing exact P-ACF invariance for the PO design over all delay-Doppler indices. In contrast, the A-ACF under the PO design and both correlation modes under the UN design may be reshaped. 

\subsection{Sensing-Subcarrier Budget and Ambiguity-Shaping Performance}\label{sec:IV-C}

The pilot-only results in Fig.~\ref{fig:pilot_design_delay_doppler_cuts} show that the prescribed pilot resources provide limited waveform-shaping freedom even after unitary optimization. We therefore reserve a subset of the original data subcarriers as dedicated sensing resources and jointly optimize the sensing and pilot blocks.

We evaluate the sensing-resource budgets $K_s\in\{0,64,128,230\}$. For each $K_s>0$, problem (P2) provides the initial sensing support, which is subsequently refined according to the complete self-ambiguity loss in problem (P2.1). The sensing symbols $\mathbf{S}_s$ and the available unitary blocks $\mathbf{U}_s$, $\mathbf{U}_c$, and $\mathbf{U}_p$ are jointly optimized under both correlation modes. Since each reserved sensing subcarrier removes one communication payload subcarrier, the ideal normalized communication payload rate is $\bar R_c=K_c/N_d=1-K_s/N_d$. The considered budgets correspond to $\bar R_c\in\{1,0.722,0.443,0\}$, respectively. This metric accounts only for resource reservation under fixed modulation and ideal channel estimation. The effect of pilot optimization on channel-estimation accuracy and the end-to-end communication rate is left for future work. The resulting ROI PSLRs are compared with conventional OFDM in Fig.~\ref{fig:pslr_sensing_budget}.
\begin{figure}[!t]
    \centering
    \subfloat[P-ACF.]{%
        \includegraphics[width=0.99\linewidth]{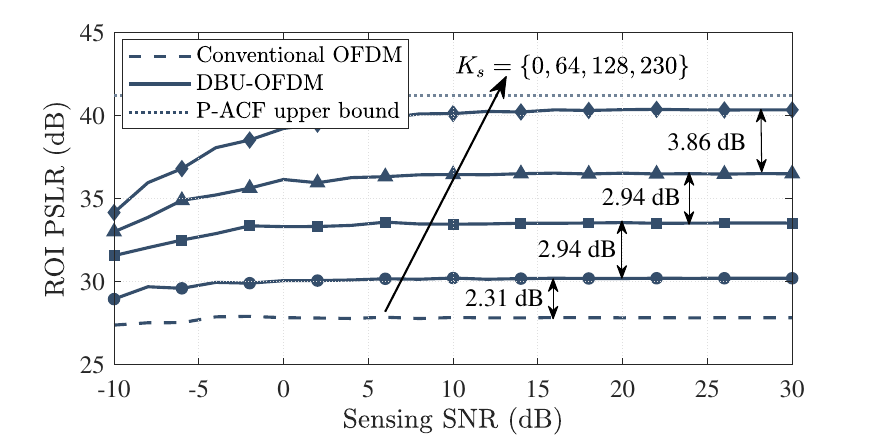}%
        \label{fig:pslr_snr_pacf}}
    \\
    \vspace{-10pt}
    \subfloat[A-ACF.]{%
        \includegraphics[width=0.99\linewidth]{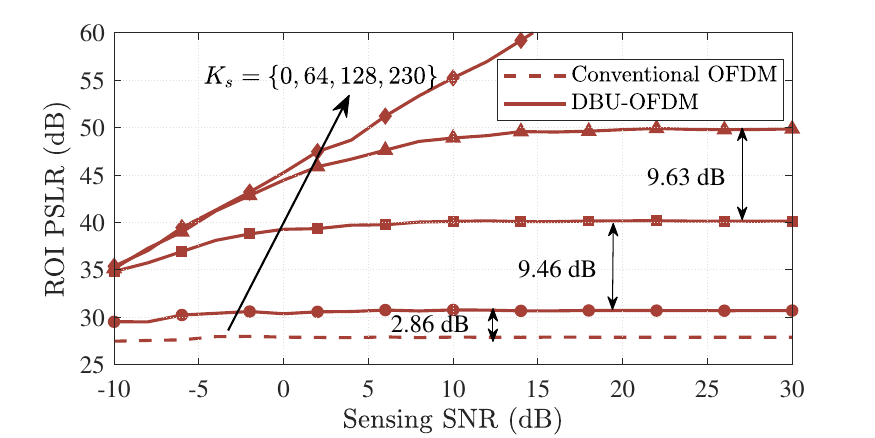}%
        \label{fig:pslr_snr_aacf}}
    \caption{ROI PSLR versus sensing SNR for DBU-OFDM with $K_s\in\{0,64,128,230\}$ under (a) the P-ACF mode and (b) the A-ACF mode. The upper bound is obtained from the relaxed P-ACF optimization problem $(\mathrm{P3})$.}
    \label{fig:pslr_sensing_budget}
\end{figure}

Fig.~\ref{fig:pslr_sensing_budget} shows that increasing $K_s$ improves the ROI PSLR under both correlation modes by enlarging the controllable waveform-shaping degrees of freedom. For most DBU-OFDM configurations, the PSLR initially increases with the sensing SNR and then gradually approaches saturation. This behavior reveals a transition from a noise-limited region to a waveform-sidelobe-limited region. At relatively low SNRs, the noise-induced correlation response limits the observed PSLR. As the SNR increases, the residual waveform sidelobes become the dominant limitation. In contrast, the conventional OFDM baselines exhibit only a weak PSLR improvement over the considered SNR range. This indicates that the prominent sidelobes induced by the practical OFDM resource structure already dominate the maximum ROI sidelobe response even when the noise power is relatively high. Therefore, directly suppressing these sidelobes is important for improving the sensing performance.

Moreover, the two correlation modes exhibit different ambiguity-shaping capabilities. Under the P-ACF mode, the improvement gradually saturates as $K_s$ increases even when all data subcarriers are reserved for sensing. To explain this behavior, we define $p_{\ell}=M^{-1}\sum_{m=0}^{M-1}|\bar S_{\ell,m}|^2$ as the average power of the $\ell$-th subcarrier. Applying DFT orthogonality to the periodic correlation gives
\begin{align}
    &\bar\Gamma_{\mathrm{P}}[k,0]
    =M\sum_{\ell=0}^{N-1}p_{\ell}e^{j2\pi\ell k/N}.
                                                                  \label{eq:pacf_zero_doppler_power}
\end{align}
The zero-Doppler P-ACF therefore depends only on the nonnegative average powers $p_{\ell}$. Increasing $K_s$ introduces additional power-allocation variables but provide limited flexibility.

To obtain a reference upper bound, we consider $K_s=N_d$ and $K_c=0$, and relax the structured waveform variables into the average-power vector $\mathbf{p}=[p_0,\ldots,p_{N-1}]^T$. The resulting problem is
\begin{subequations}\label{eq:pacf_power_upper_bound_problem}
\begin{align}
    (\mathrm{P3}):\,
    &\underset{\mathbf{p},t}{\operatorname{min}}
      \quad t
      \tag{\theparentequation}\label{eq:pacf_power_upper_bound_objective}\\
    &\operatorname{s.t.}\quad
      \left|\sum_{\ell=0}^{N-1}
      p_{\ell}e^{j2\pi\ell k/N}\right|\leq t,
      \quad k\in\Omega_k
      \tag{\theparentequation a}\label{eq:pacf_power_upper_bound_sidelobe}\\
    &\hphantom{\operatorname{s.t.}\quad}
      \sum_{\ell\in\bm{\mathcal{K}}_s}p_{\ell}=K_s,
      \quad
      \sum_{\ell\in\bm{\mathcal{K}}_p}p_{\ell}=N_p
      \tag{\theparentequation b}\label{eq:pacf_power_upper_bound_energy}\\
    &\hphantom{\operatorname{s.t.}\quad}
      p_{\ell}\geq0,\ \ell\in\bm{\mathcal{K}}_s\cup\bm{\mathcal{K}}_p,
      \,
      p_{\ell}=0,\ \ell\in\bm{\mathcal{K}}_0.
      \tag{\theparentequation c}\label{eq:pacf_power_upper_bound_support}
\end{align}
\end{subequations}
Problem (P3) retains only the zero-Doppler delay slice and allows arbitrary nonnegative power distributions under the total-power constraints. It therefore provides an optimistic upper bound for the complete DBU-OFDM design. Under the parameters in Table~\ref{tab:numerical_parameters} with $K_s=230$, the solution of problem (P3) gives a P-ACF PSLR upper bound of $41.25$ dB, as shown in Fig.~\ref{fig:pslr_snr_pacf}. 

In contrast, the finite summation interval in the A-ACF definition in Eq. \eqref{eq:a_acf_1d} retains the cross-subcarrier terms. Both the amplitudes and phases of the sensing and pilot symbols therefore remain available for sidelobe shaping. This greater design freedom explains the substantially greater A-ACF improvement in Fig.~\ref{fig:pslr_snr_aacf}, where the $K_s=230$ curve continues to increase approximately linearly over the evaluated SNR range.

The corresponding self-ambiguity surfaces for $K_s=128$ and $K_s=230$ are compared in Fig.~\ref{fig:ambiguity_surfaces_sensing_budget}. These surface views further illustrate how increasing the sensing-resource budget suppresses the dominant sidelobes over the selected delay-Doppler region.

\begin{figure}[!t]
    \centering
    \subfloat[P-ACF, $K_s=128$.]{%
        \includegraphics[width=0.48\linewidth]{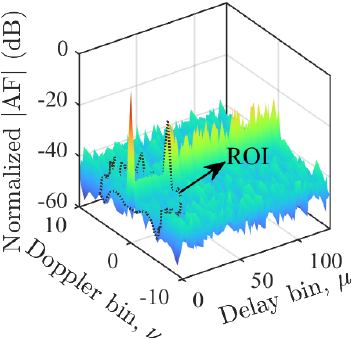}%
        \label{fig:ambiguity_surface_pacf_ks128}}
    \hfill
    \subfloat[P-ACF, $K_s=230$.]{%
        \includegraphics[width=0.48\linewidth]{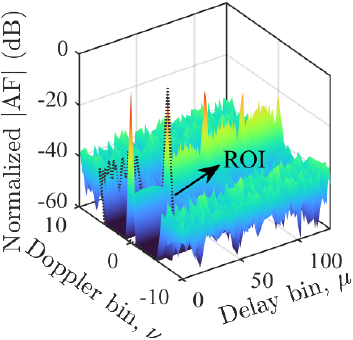}%
        \label{fig:ambiguity_surface_pacf_ks230}}
    \\
    \vspace{-8pt}
    \subfloat[A-ACF, $K_s=128$.]{%
        \includegraphics[width=0.48\linewidth]{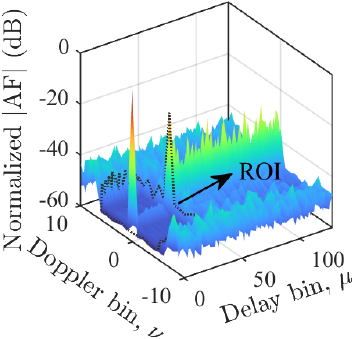}%
        \label{fig:ambiguity_surface_aacf_ks128}}
    \hfill
    \subfloat[A-ACF, $K_s=230$.]{%
        \includegraphics[width=0.48\linewidth]{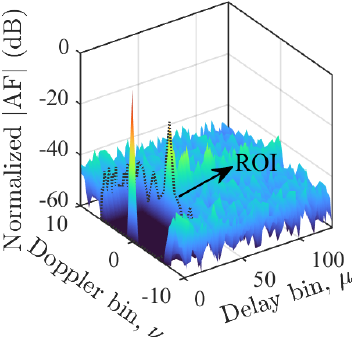}%
        \label{fig:ambiguity_surface_aacf_ks230}}
    \caption{Normalized two-dimensional self-ambiguity functions of the optimized DBU-OFDM waveforms for $K_s=128$ and $K_s=230$. (a) and (b) show the P-ACF results, while panels (c) and (d) show the A-ACF results.}
    \label{fig:ambiguity_surfaces_sensing_budget}
\end{figure}

Fig.~\ref{fig:ambiguity_surfaces_sensing_budget} confirms that increasing $K_s$ provides limited additional P-ACF suppression along the zero-Doppler delay slice, whereas the A-ACF mode yields substantially lower sidelobes throughout the selected ROI, particularly when $K_s=230$. The optimized surfaces also show localized suppression, with some sidelobes outside the ROI becoming more pronounced as the dominant ROI sidelobes decrease. Thus, this sidelobe redistribution is consistent with the observation in Section~\ref{sec:IV-B} that ROI-oriented training reshapes the sidelobes locally rather than uniformly suppressing them over the complete delay-Doppler plane.


\subsection{Ablation Study and Learned Block Structure}\label{sec:IV-D}

We further conduct an A-ACF ablation study at $K_s=128$ and $K_c=102$. Table~\ref{tab:aacf_ablation} compares the noiseless ROI PSLRs for three support initializations. The fixed-support configurations retain the selected subcarrier locations during waveform training, whereas the refined-support configurations further optimize these locations using the complete ROI loss.

\begin{table}[!t]
    \centering
    \caption{A-ACF ablation results with $K_s=128$ and $K_c=102$.}
    \label{tab:aacf_ablation}
    \renewcommand{\arraystretch}{1.08}
    \setlength{\tabcolsep}{3pt}
    \footnotesize
    \begin{tabular}{@{}llcrr@{}}
        \toprule
        Initialization & Support & $\mathbf{U}_c$ & PSLR (dB) & Gain (dB) \\
        \midrule
        Conventional OFDM & -- & -- & 27.93 & 0.00 \\
        \midrule
        Array-factor & Fixed & Identity & 49.50 & +21.57 \\
        Uniform & Fixed & Identity & 37.88 & +9.95 \\
        Coprime & Fixed & Identity & 37.94 & +10.01 \\
        \midrule
        Array-factor & Refined & Identity & 51.01 & +23.08 \\
        Uniform & Refined & Identity & 51.06 & +23.13 \\
        Coprime & Refined & Identity & 51.11 & +23.18 \\
        Array-factor & Refined & Trainable & 50.97 & +23.03 \\
        \bottomrule
    \end{tabular}
\end{table}

With fixed supports, array-factor initialization obtained via $\left(\text{P2}\right)$ achieves a PSLR of $49.50$ dB, exceeding the better of the uniform and coprime results by $11.56$ dB. After support refinement, the three initializations achieve PSLRs between $51.01$ and $51.11$ dB. These results show that array-factor initialization provides an effective starting support, while subsequent refinement reduces the sensitivity to initialization.

Moreover, the array-factor refined configurations with fixed and trainable $\mathbf{U}_c$ achieve similar PSLRs. To examine this results, Fig.~\ref{fig:aacf_unitary_blocks} shows $\mathbf{U}_c$ and $\mathbf{U}_s$, respectively. Applying Eq. \eqref{eq:monomiality_ratio} with dimensions $K_c$ and $K_s$ gives $\rho_{\mathrm{mon}}(\mathbf{U}_c)=0.9793$ and $\rho_{\mathrm{mon}}(\mathbf{U}_s)=0.0863$, respectively. Thus, we also have $\mathbf{U}_c\approx \mathbf{D}_c\mathbf{P}_c$, where $\mathbf{D}_c$ is a diagonal phase matrix and $\mathbf{P}_c$ is a permutation matrix. It mainly reorders and rotates communication symbols with limited mixing, whereas $\mathbf{U}_s$ retains a dense mixing structure.

\begin{figure}[!t]
    \centering
    \subfloat[$\mathbf{U}_c$, $K_c=102$.]{%
        \includegraphics[width=0.43\linewidth]{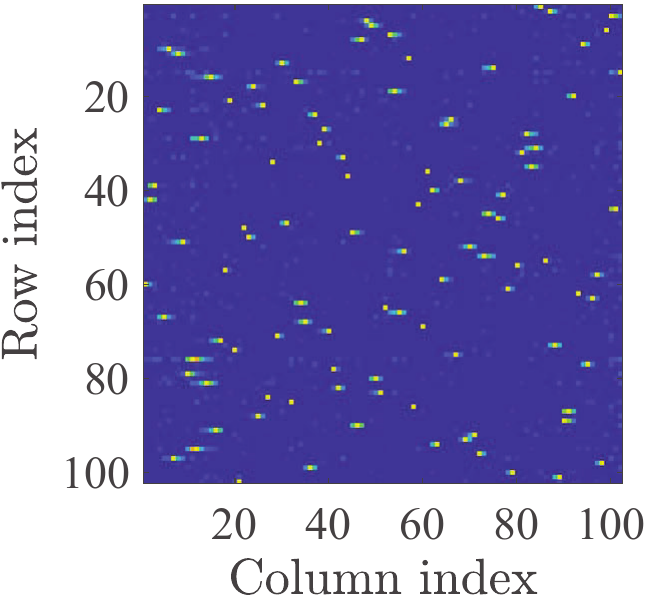}%
        \label{fig:aacf_communication_block}}
    \hfill
    \subfloat[$\mathbf{U}_s$, $K_s=128$.]{%
        \includegraphics[width=0.53\linewidth]{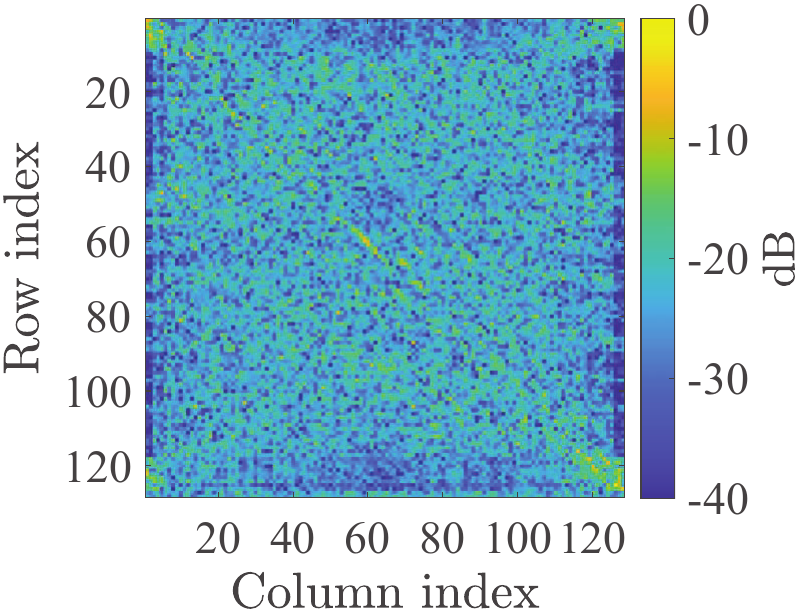}%
        \label{fig:aacf_sensing_block}}
    \caption{Learned communication and sensing blocks for the A-ACF array-factor refined configuration with trainable $\mathbf{U}_c$. Their monomiality ratios are $0.9793$ and $0.0863$, respectively.}
    \label{fig:aacf_unitary_blocks}
\end{figure}

Taken together, these observations support $\mathbf{U}_c=\mathbf{I}_{K_c}$ as a practical simplification under the considered configuration, preserving the original communication-symbol mapping while retaining effective ROI sidelobe suppression through the sensing and pilot blocks. However, pilot optimization may still affect channel estimation. A rigorous characterization of this design choice and its end-to-end communication impact remains for future work due to the limitation of space.

\section{Conclusions}\label{sec:conclusion}

In this paper, we developed ROI-oriented DBU-OFDM to suppress sidelobes in practical OFDM frames while preserving resource isolation, energy, invertibility, and DFT-based communication processing under the stated conditions. We combined unitary waveform shaping with sensing-support allocation and established the representation capability and invariance properties that characterize the available design freedom.

Under full resource occupancy, our experiments corroborated the established P-ACF optimality and supported the conjectured global optimality of OFDM for the considered zero-Doppler A-ACF delay-slice PSLR problem within the unitary waveform class. With practical resources, pilot optimization improved the ROI PSLR by over $2$ dB, while dedicated sensing subcarriers yielded gains of several to tens of dB. Moreover, A-ACF offered greater shaping freedom than P-ACF under the considered configuration.

Future work will jointly evaluate pilot optimization, channel-estimation accuracy, and achievable communication rate and extend the formulation to fractional delay-Doppler parameters and time-varying channels with non-negligible intrasymbol Doppler.

{\appendices
\section{Proof of Lemma \ref{lemma:1}}\label{app:A}
    This result is derived from the standard Householder decomposition in \cite{ivanov2006engineering}. We include its recursive construction to explain why $N_{\vartheta}-1$ reflection stages are sufficient and to establish the correspondence with Eq. \eqref{eq:householder_parameterization}.

Consider a unit-norm vector $\mathbf{a}\in\mathbb{C}^{r\times1}$, its first entry is denoted by $a_1=[\mathbf{a}]_1$, and $\mathbf{e}_{1}=[1,0,\ldots,0]^{T}$ denotes the coordinate vector whose first entry is one. We select a unit-modulus scalar $\eta$ such that $\eta^{*}a_1=-|a_1|$. If $a_1=0$, any unit-modulus $\eta$ can be used. Then, the reflection vector is defined as
\begin{align}
    &\mathbf{v}=\mathbf{a}-\eta\mathbf{e}_{1}.               \label{eq:householder_reflection_vector}
\end{align}
The corresponding Householder matrix is
\begin{align}
    &\widehat{\mathbf{H}}
     =\mathbf{I}_{r}
      -2\frac{\mathbf{v}\mathbf{v}^{H}}
               {\mathbf{v}^{H}\mathbf{v}}.                    \label{eq:householder_local_reflection}
\end{align}
Since $\mathbf{a}^{H}\mathbf{a}=1$ and $\eta^{*}a_1=-|a_1|$, the inner products required to evaluate the reflection are
\begin{align}
    &\mathbf{v}^{H}\mathbf{a}=1+|a_1|,\qquad
     \mathbf{v}^{H}\mathbf{v}=2(1+|a_1|).                    \label{eq:householder_inner_products}
\end{align}
Thus we have $\mathbf{v}^{H}\mathbf{v}\neq 0$. Substituting Eq. \eqref{eq:householder_inner_products} into Eq. \eqref{eq:householder_local_reflection} gives
\begin{align}
    &\widehat{\mathbf{H}}\mathbf{a}
     =\mathbf{a}-2\mathbf{v}
      \frac{\mathbf{v}^{H}\mathbf{a}}
           {\mathbf{v}^{H}\mathbf{v}}
     =\mathbf{a}-\mathbf{v}
     =\eta\mathbf{e}_{1}.                                    \label{eq:householder_vector_mapping}
\end{align}
Thus, a standard Householder reflection maps any unit-norm vector to a vector whose only nonzero entry is the first element.

We initialize the decomposition with $\mathbf{U}_{\vartheta}^{(0)}=\mathbf{U}_{\vartheta}^{H}$. Applying the construction in Eq. \eqref{eq:householder_vector_mapping} to the first column of $\mathbf{U}_{\vartheta}^{(0)}$ gives
\begin{align}
    &\mathbf{H}_{\vartheta,1}\mathbf{U}_{\vartheta}^{(0)}
     =\begin{bmatrix}
       \eta_{\vartheta,1}
       &\mathbf{0}_{1\times(N_{\vartheta}-1)}\\
       \mathbf{0}_{(N_{\vartheta}-1)\times1}
       &\mathbf{U}_{\vartheta}^{(1)}
      \end{bmatrix},                                        \label{eq:householder_recursive_block}
\end{align}
where $\mathbf{U}_{\vartheta}^{(1)}$ denotes the $(N_{\vartheta}-1)\times(N_{\vartheta}-1)$ trailing block. Since both $\mathbf{H}_{\vartheta,1}$ and $\mathbf{U}_{\vartheta}^{(0)}$ are unitary, this trailing block is also unitary. The same construction is then applied recursively to the trailing block, with $\mathbf{U}_{\vartheta}^{(i)}$ denoting the unitary block remaining after the $i$-th reflection. Each reflection isolates one unit-modulus diagonal element and reduces the dimension of the trailing unitary block by one.

After $N_{\vartheta}-1$ steps, $\mathbf{U}_{\vartheta}^{(N_{\vartheta}-1)}$ is a $1\times1$ unitary matrix and therefore equals a unit-modulus scalar $\eta_{\vartheta,N_{\vartheta}}$. Consequently, the recursive decomposition produces the diagonal phase matrix
\begin{align}
    &\mathbf{H}_{\vartheta,N_{\vartheta}-1}\cdots
     \mathbf{H}_{\vartheta,1}\mathbf{U}_{\vartheta}^{H}
     =\mathbf{\Phi}_{\vartheta},                             \label{eq:householder_diagonalization}
\end{align}
where $\mathbf{\Phi}_{\vartheta}=\operatorname{diag}(\eta_{\vartheta,1},\ldots,\eta_{\vartheta,N_{\vartheta}})$ is a diagonal phase matrix. Taking the Hermitian transpose of Eq. \eqref{eq:householder_diagonalization} and using the Hermitian and involutory properties of each Householder matrix yields
\begin{align}
    &\mathbf{U}_{\vartheta}
     =\mathbf{\Phi}_{\vartheta}^{H}
      \mathbf{H}_{\vartheta,N_{\vartheta}-1}\cdots
      \mathbf{H}_{\vartheta,1},                              \label{eq:householder_universality}
\end{align}
 Therefore, setting $\mathbf{D}_{\vartheta}=\mathbf{\Phi}_{\vartheta}^{H}$ gives the parameterization in Eq. \eqref{eq:householder_parameterization}, and a depth budget satisfying $K_{\vartheta}\geq N_{\vartheta}-1$ is sufficient for complete unitary representation.

 \section{Proof of Lemma \ref{lemma:lemma2}}\label{app:B}
 Under the phase-only pilot design, the frequency-domain symbols satisfy $\bar S_{\ell,m}=e^{j\theta_{\ell,m}}S_{\ell,m}$ for $\ell\in\bm{\mathcal K}_{p}$ and remain unchanged on the data and null subcarriers.

Substituting the normalized IDFT in Eq. \eqref{eq:ofdm_modulation} into the P-ACF definition in Eq. \eqref{eq:p_acf_1d} gives
\begin{align}
    &r_{\mathrm{P},m}[k]
    =\frac{1}{N}
      \sum_{\ell=0}^{N-1}\sum_{p=0}^{N-1}
      S_{\ell,m}^{*}S_{p,m}
      e^{j2\pi pk/N}
      \sum_{n=0}^{N-1}e^{j2\pi(p-\ell)n/N}.                 \label{eq:pacf_idft_substitution}
\end{align}
The inner summation represents the orthogonality relation between two DFT basis vectors and satisfies
\begin{align}
    &\sum_{n=0}^{N-1}e^{j2\pi(p-\ell)n/N}
    =
    \begin{cases}
        N, &p=\ell,\\
        0, &p\neq\ell.
    \end{cases}                                               \label{eq:dft_basis_orthogonality}
\end{align}
Therefore, the P-ACF can be equivalently expressed as
\begin{align}
    &r_{\mathrm{P},m}[k]
    =\sum_{\ell=0}^{N-1}|S_{\ell,m}|^2
      e^{j2\pi\ell k/N}.                                    \label{eq:pacf_frequency_representation}
\end{align}

Since the phase-only adjustment preserves the magnitude of every frequency-domain symbol, we have
\begin{align}
    &|\bar S_{\ell,m}|^2
    =\left|e^{j\theta_{\ell,m}}S_{\ell,m}\right|^2
    =|S_{\ell,m}|^2,
    \, \ell\in\mathbb{Z}_N,\, m\in\mathbb{Z}_M.      \label{eq:phase_only_magnitude_invariance}
\end{align}
Substituting this equality into the frequency-domain P-ACF in Eq. \eqref{eq:pacf_frequency_representation} gives
\begin{align}
    &\bar r_{\mathrm{P},m}[k]
    =\sum_{\ell=0}^{N-1}|\bar S_{\ell,m}|^2
      e^{j2\pi\ell k/N}
      \nonumber\\
    &\phantom{\bar r_{\mathrm{P},m}[k]}
    =\sum_{\ell=0}^{N-1}|S_{\ell,m}|^2e^{j2\pi\ell k/N}
    =r_{\mathrm{P},m}[k],\quad k\in\mathbb{Z}_N.           \label{eq:phase_only_pacf_invariance}
\end{align}
Substituting this per-symbol equality into Eq. \eqref{eq:self_af} gives
\begin{align}
    &\bar\Gamma_{\mathrm{P}}[k,q]
    =\sum_{m=0}^{M-1}\bar r_{\mathrm{P},m}[k]e^{-j2\pi qm/M}
    \nonumber\\
    &\phantom{\bar\Gamma_{\mathrm{P}}[k,q]}
    =\sum_{m=0}^{M-1}r_{\mathrm{P},m}[k]e^{-j2\pi qm/M}
    =\Gamma_{\mathrm{P}}[k,q].                              \label{eq:phase_only_pacf_af_invariance}
\end{align}
Therefore, the complete P-ACF self-ambiguity function, including all its sidelobes, remains unchanged.

 \section{Proof of Lemma \ref{lemma:lemma3}}\label{app:C}
At $k=0$, the P-ACF and A-ACF definitions both reduce to
\begin{align}
    &r_{\chi,m}[0]
    =\sum_{n=0}^{N-1}|x_{n,m}|^2
    =\|\mathbf{x}_m\|_2^2
    =\|\mathbf{s}_m\|_2^2,
    \quad \chi\in\{\mathrm{P},\mathrm{A}\}.               \label{eq:zero_delay_symbol_energy}
\end{align}
For an arbitrary unitary DBU-OFDM transformation, the corresponding zero-delay correlation value satisfies
\begin{align}
    \bar r_{\chi,m}[0]
    &=\|\bar{\mathbf{x}}_m\|_2^2
    =\|\mathbf{F}_N^H\mathbf{U}\mathbf{s}_m\|_2^2
    =\|\mathbf{U}\mathbf{s}_m\|_2^2
    =\|\mathbf{s}_m\|_2^2\nonumber\\
    &=r_{\chi,m}[0]. \label{eq:dbu_symbol_energy_invariance}
\end{align}
Therefore, DBU-OFDM and conventional OFDM have the same per-symbol energy sequence. Since the zero-delay Doppler cut in Eq. \eqref{eq:self_af} is the DFT of this sequence, Eq. \eqref{eq:zero_delay_invariance} follows.

}
\bibliography{Reference}
\end{document}